\documentclass[a4paper,11pt,twoside]{article}

\usepackage[utf8]{inputenc}
\usepackage[T1]{fontenc}
\usepackage[english]{babel}

\usepackage{charter}

\usepackage{indentfirst}
\usepackage{xcolor}
\usepackage{graphicx}
\usepackage{verbatim}

\usepackage{caption,subcaption} 
\usepackage{multirow}
\usepackage{booktabs}
\usepackage{hyperref}

\usepackage{pifont}
\newcommand{\cmark}{\textrm{\ding{52}}}%
\newcommand{\xmark}{\textrm{\ding{56}}}%

\usepackage[top=3.cm, bottom=4.0cm, left=2.2cm, right=2.2cm]{geometry}
\usepackage{newtxtext,newtxmath,amsmath}

\usepackage{amsthm}	
\usepackage{amsfonts}	

\usepackage{amssymb	
           ,bbm
            ,units 
           }

\usepackage{enumerate}

\usepackage[square,numbers,sort&compress]{natbib}
\usepackage{
            ulem		
           ,soul		
} 

\numberwithin{equation}{section}
\numberwithin{figure}{section}
\numberwithin{table}{section}
 \usepackage[nodayofweek]{datetime}

\title{Forward utility and market adjustments in relative investment-consumption games of many players}
\author{
 	\normalsize Gon\c calo dos Reis\footnote{G.~dos Reis acknowledges support from the \emph{Funda{\c c}$\tilde{\text{a}}$o para a Ci$\hat{e}$ncia e a Tecnologia} (Portuguese Foundation for Science and Technology) through the project UIDB/00297/2020 (Centro de Matem\'atica e Aplica\c c$\tilde{\text{o}}$es CMA/FCT/UNL)} \\[8pt]
           \footnotesize University of Edinburgh\\
           \footnotesize School of Mathematics \\
           \footnotesize Edinburgh, EH9 3FD, UK\\  
           \footnotesize and \\
	         \footnotesize Centro de Matem\'atica e Aplica\c c$\tilde{\text{o}}$es \\
	       \footnotesize (CMA), FCT, UNL, Portugal \\
          \footnotesize  G.dosReis@ed.ac.uk
 \and
\normalsize Vadim Platonov\\[8pt] 
         \footnotesize University of Edinburgh\\ 
         \footnotesize   School of Mathematics \\
         \footnotesize   Edinburgh, EH9 3FD, UK\\    
         \footnotesize     \\
		     \footnotesize   \\
         \footnotesize         \\
        \footnotesize  v.d.platonov@sms.ed.ac.uk
}

\date{ \currenttime, \ddmmyyyydate\today\qquad{(File: \tt \jobname.tex})}

\theoremstyle{plain}
\newtheorem{theorem}{Theorem}[section]

\newtheorem{proposition}[theorem]{Proposition}
\newtheorem{corollary}[theorem]{Corollary}

\newtheorem{definition}[theorem]{Definition}

\newtheorem{remark}[theorem]{Remark}

\newtheorem{assumption}[theorem]{Assumption}

\newcommand{\bE}{\mathbb{E}}
\newcommand{\bF}{\mathbb{F}}

\newcommand{\bP}{\mathbb{P}}

\newcommand{\bR}{\mathbb{R}}

\newcommand{\cA}{\mathcal{A}}

\newcommand{\cF}{\mathcal{F}}

\newcommand{\cZ}{\mathcal{Z}}

\definecolor{darkgreen}{rgb}{0,0.35,0}

\newcommand{\geqs}{\geqslant}
\newcommand{\leqs}{\leqslant}

\begin{document}

\maketitle

\begin{abstract}
    We study a portfolio management problem featuring many-player and mean field competition, investment and consumption, and relative performance concerns under the forward performance processes (FPP) framework. We focus on agents using power (CRRA) type FPPs for their investment-consumption optimization problem under a common noise Merton market model. We solve both the many-player and mean field game providing closed-form expressions for the solutions where the limit of the former yields the latter. \\
    In our case, the FPP framework yields a continuum of solutions for the consumption component as indexed to a market parameter we coin ``market-risk relative consumption preference''. The parameter permits the agent to set a preference for their consumption going forward in time that, in the competition case, reflects a common market behaviour.
    We show the FPP framework, under both competition and no-competition, allows the agent to disentangle her risk-tolerance and elasticity of intertemporal substitution (EIS) just like Epstein-Zin preferences under recursive utility framework and unlike the classical utility theory one. This, in turn, allows a finer analysis on the agent's consumption ``income'' and ``substitution'' regimes, and, of independent interest,  motivates a new strand of economics research on EIS under the FPP framework.\\
    We find that competition rescales the agent's perception of consumption in a non-trivial manner. We provide numerical illustrations of our results. 
\end{abstract}

\vspace{0.2cm}

\vspace{0.2cm}



 \newpage
\section{Introduction}

We study many-player games of investment-consumption optimisation under the forward performance framework in a Merton market model featuring common-noise. In each game, the agents trade only in a specific stock affected by a common market signal and seek to optimise their wealth and consumption process while having relative concerns towards the average wealth and consumption of the other agents. The core features of our model are (a) relative consumption concerns, (b) relative wealth concerns, (c) asset specialisation, (d) game competition (finite-player and mean field games) and (e) the forward performance processes (FPP) view. The former (a)--(d) have been analysed in  \cite{lackersoret2020many} through the lens of the classical utility framework. In this work we subscribe to the elegantly argued (and supported by empirical evidence) FPP paradigm of \cite{AnthropelosGengZariphopoulou2020} and view  \cite{lackersoret2020many} through its light. The arguments we make using FPPs, competition and the presence of consumption make the analysis involved and reveal elements not present in \cite{lackersoret2020many} (despite the similarity to (a)--(d)) or \cite{karatzas1997explicit}.

Throughout we consider fund managers that trade between their individual stock and common riskless asset, the so-called problem of \emph{asset specialisation}.  This problem was addressed initially by \cite{brennan1975optimal,merton1969lifetime} and further developed in \cite{van2009information},   \cite{kacperczyk2005industry},  \cite{LackerZariphopoulou2017, lackersoret2020many},  \cite{AnthropelosGengZariphopoulou2020}.
These works report on asset specialisation for various reasons: asset familiarity, trading costs and constraints, reduction of learning costs or industry specialisation. Competition games of finite-player or mean field type in the context of asset specialisation has received much attention recently, and we defer to  \cite{LackerZariphopoulou2017,lackersoret2020many,platonov2020forward,AnthropelosGengZariphopoulou2020} for an overview that relates to our context. For a long-view of mean field games theory and its applications we refer to the monographs \cite{carmona2018probabilistic}.

The literature on portfolio management for agents with utility preferences under performance concerns is a growing one. Benchmarking is a feature of human nature and critical for fund managers who need to keep the fund competitive. We refer to the excellent economic and finance motivations found in \cite{AnthropelosGengZariphopoulou2020} and also  \cite{BielagkLionnetDosReis2017,LackerZariphopoulou2017,lackersoret2020many}. In this work, we build on the structure proposed firstly in \cite{EspinosaTouzi2015} and then in \cite{FreiDosReis2011} and  \cite{BielagkLionnetDosReis2017,LackerZariphopoulou2017,deng2020relative,lackersoret2020many}. Additionally, we point the reader to the beautiful introductions of \cite{LackerZariphopoulou2017,lackersoret2020many} who brought those concepts to the framework of mean field games. Further, those works also make for an excellent review of mean field games in the context of the Merton problem. Investment under performance concerns has been taken up in many variants, most recently: \cite{deng2020relative} analyse a partial information finite-player CARA utility game with performance concerns employing Forward Backward SDE (FBSDE) machinery (market model with stochastic coefficients), see \cite[Table 1]{deng2020relative}. \cite{fu2020mean} works in the same context but addresses the mean field game in the full non-Markovian framework using mean field FBSDEs, see also \cite[Table 1]{fu2020mean}. Aspects of investment under relative consumption are much less explored, only in \cite{lackersoret2020many}.

At the core of our work are the ideas underpinning the \textit{forward investment performance criteria} introduced by Musiela and Zariphopoulou \cite{musiela2006investments} and  \cite{henderson2007horizon} as a way to solve portfolio optimisation problems without the specific drawbacks of the classical utility theory. 
In classical utility theory when entering the market, investors prescribe their risk profile at horizon time and therefore cannot adapt it to changes in market conditions or update risk preferences; additionally, the investment time horizon is fixed, and the portfolio is derived in respect to this temporal reference point. 
In opposition, the forward criteria (portfolio) optimisation problem is a maximisation of a conditional expectation of a certain \textit{stochastic utility function}. Here, investors need only to state their risk profile for the market entry time $t=0$. 
Hence, time-horizons are arbitrary which is more realistic for practitioners. In essence, the forward \textit{dynamic} performance map is built with the property of dynamic consistency taken as the starting point.

\textit{Forward performance processes} (FPP), as defined in Musiela and Zariphopoulou
\cite{MusielaZariphopoulou2009}, capture the time evolution of such stochastic utility functions and since then much progress has been done towards their characterisation in a variety of settings \cite{berrier2011forward}, \cite{KarouiMrad2013}, \cite{Zitkovic2009},
\cite{elkaroui2018consistent}, 
\cite{KarouiHillairetMrad2019}, 
\cite{chong2018optimal}, 
\cite{kallblad2020black}, 
\cite{li2020game}. 
For excellent literature overviews on the developments of FPP we point the reader to \cite{chong2019pricing}, \cite{avanesyan2020construction}, \cite{matoussi2020dynamic} and \cite{kallblad2020black}. 
The latter provides a near taxonomic breakdown of the literature of the forward criteria with respect to their behaviour in time along with approaches used. Here, as in \cite{kallblad2020black} or \cite{avanesyan2020construction}, we focus on the class of forward investment-consumption criteria with the specific property that they are differentiable in time (a dynamics without volatility). 
In \cite{avanesyan2020construction} the authors study an optimal portfolio selection problem (no consumption, no competition) under the FPP of power form in an incomplete market. Their arguments make use of HJB machinery and, like in this work or \cite{chong2018optimal}, an assumption of ``separable power factor form'' for the FPP is made (Assumption \ref{ass:form-of-U-V-mean-field-g-with-f} below). 
The study of consumption with FPPs (or forward utility) was considered in \cite{berrier2011forward,elkaroui2018consistent,chong2018optimal}; interesting is \cite{chong2018optimal} who manage the difficulty associated to having a FPP dynamics featuring a volatility component by assuming the volatility to be exogenously postulated by the agent as a market perception. Their FPP consumption optimisation problem yields non-uniqueness.

Competition within the framework of FPPs is in its infancy. The concept originated in \cite{geng2017passive} for a two-player game in the CRRA context with a Merton market model. Its vision is expanded upon in \cite{AnthropelosGengZariphopoulou2020}, still within a two-player game but allowing for random coefficients in an incomplete market model. 
Inspired by \cite{geng2017passive} and \cite{LackerZariphopoulou2017}, the formulation of \textit{mean field forward performance games} and the concept of \textit{mean field forward equilibrium} was proposed by \cite{platonov2020forward}. 
From a competition perspective, the FPP approach reflects that agents need not all have the same time-horizons since, under asset specialisation, different industries have different timeframes (although quarterly reporting are common points of reference), see \cite{AnthropelosGengZariphopoulou2020} for a full discussion. 

The FPP framework is not the only theory seeking to overcome the limitations of the classical utility theory. An alternative from the 90s is the so-called Epstein-Zin preferences within the recursive utility framework \cite{epstein1989substitution,weil1989equity} and expanding on the theoretical framework of Kreps-Porteus \cite{krepsporteus1978}. 
For the consumption problem, the classic CRRA utility optimisation yields a strict relation between the agent's risk tolerance $\delta$ and \textit{Elasticity of Intertemporal Substitution} (EIS), namely that $\textrm{EIS}^{\textrm{classic, \cite{lackersoret2020many}}}=\delta$ 
(e.g., see \cite{lackersoret2020many}) and disentangling both was at the heart of the Epstein-Zin preferences.  
As discussed in \cite{thimme2017intertemporal} regarding EIS, higher interest rates increase the overall wealth of the consumer due to higher cash-flows in future periods. 
The effect where consumers spend a part of this higher future income today is called ``income effect''. 
On the other hand, with higher interest rates, one needs to save
a smaller fraction of today's consumption to have an additional consumption unit tomorrow. 
This motivation to save more today and postpone today's consumption is called the ``substitution effect''. 
Consumers with a high EIS are more willing to substitute consumption over time, which directly impacts the ``substitution effect''. Now, classic theory yields $\textrm{EIS}^{\textrm{classic, \cite{lackersoret2020many}}}=\delta$. However, the risk tolerance is atemporal, relating how a consumer substitutes consumption across different states of the world. In opposition, EIS is intertemporal, relating how a consumer substitutes consumption between now and later \cite{cundy2018essays}. Thus the classic utility framework cannot capture how agent competition, with or without performance concerns, changes EIS. 

Time-continuous stochastic preferences capturing the \textit{intertemporal substitution} of Hindy-Huang-Kreps \cite{hindy1992intertemporal} have been used to understand consumption optimisation \cite{benth2001optimal} and Arrow-Debreu equilibria \cite{bank2001existence}. More recently, \cite{aase2016recursive} revisited the EIS discussion using \cite{duffie1992stochastic}'s stochastic recursive utility in continuous time specified to the Kreps-Porteus \cite{krepsporteus1978} family. Forthcoming is \cite{bismuthgueant2019} who recast the MFG of \cite{LackerZariphopoulou2017} under the Epstein-Zin/Kreps-Porteus recursive utility framework. Nonetheless, inspecting this reference one sees several of the criticisms aimed at the standard utility framework appearing again (see \cite{AnthropelosGengZariphopoulou2020}): \cite{campani2019approximate} emphasizes the dependence on investment's horizon (``The investor's horizon also plays a crucial role in optimal policies'') and that the underlying model is fixed throughout the investment time frame. In a nutshell, the FPP works with the forward point of view to investment (risk profile is prescribed at $t=0$) while the recursive utility still works within the backward one (risk profile is prescribed at a horizon time $T>0$.

With this work, we investigate an $n$-player and mean field game for asset specialised agents who optimise their investment-consumption under relative wealth-consumption concerns through the lens of the FPP framework. The tractability of our setting yields findings not seen in the classic utility theory. Namely, \textbf{our contributions are:}

    (I) We expand the concept of CRRA FPP to a game-theoretical framework both for finite-player games and mean field games within the common-noise Merton market model. The MFG approach is based on the simpler concept work \cite{platonov2020forward} for CARA FPP (wealth optimisation only), and \cite{lackersoret2020many} for the classical CRRA investment-consumption utility problem. For both games, we provide explicit expressions for the quantities of interest. From a methodological point of view, we solve the control problem by combining convex duality with HJB-type arguments.
    
    (II) With the inclusion of consumption, the CRRA FPP wealth-consumption problem does not have a unique solution but a class of them indexed by a specific parameter $\kappa\in\bR$ interpreted as \emph{market-risk relative consumption preference}. 
This parameter appears in both game settings and in the single-agent optimisation problem (no competition/performance concerns). Under competition, $\kappa$ is common to all players and reflects how the environment weights the utility from consumption in respect to the one from wealth. 
    The tractability of our setting, exposes the $\kappa$ parameter to clearly enable an extra layer of interpretable market modelling (see Section \ref{sec:examplesInterpretations}) and, critically, when $\kappa=0$ we recover classic utility theory results  \cite{lackersoret2020many} also \cite{karatzas1997explicit}.

    (III) We show that the CRRA FPP framework encapsulates the same key feature of the Epstein-Zin preferences recursive utility even without competition. 
It breaks the strict relation between risk tolerance $\delta$ and \textit{Elasticity of Intertemporal Substitution} (EIS) of the classical utility theory and at the same time keeps its core features of independence from the investment time-horizon and flexibility towards updating risk preferences. 
Namely, $\textrm{EIS}^{\textrm{no competition} (\theta = 0)} = \kappa \delta$, with $\kappa$ spanning another dimension of the agent's risk preference.
    
    (IV) We show that performance concerns ($\theta\in(0,1]$) re-scale the agent's perception of consumption. Namely, for $t\geqs 0$
    \begin{align*}
                \textrm{EIS}^{\textrm{with competition} (\theta\neq 0)}_t 
                = P_t\times {\textrm{EIS}^{\textrm{no competition} (\theta=0)}}, 
                \ \textrm{ where }\ \textrm{EIS}^{\textrm{no competition} (\theta=0)}=\kappa\delta,
    \end{align*}
    where $P_t$ is a random stochastic process depending on risk-competition preferences, the market-risk relative consumption preference $\kappa$ (that is uniform for the entire population of competitors) and equilibrium consumption.  The  time-dependence of $P$ is intimately related to the use of the FPP framework with competition and discussed in detail in Section \ref{sec:examplesInterpretations} (see Equation \eqref{eq:EIS-for-our-setting}). Similar results using Epstein-Zin's recursive utility exist \cite{bismuthgueant2019}.

\smallskip
\textbf{Organisation of the paper.}  
In Section \ref{sec2} we introduce the financial market and notation. In Sections \ref{sec:FRPCwithaverage} and \ref{sec:MFG-forward-consumption} we study the finite-player and mean field game respectively. The interpretations of equilibria, parameter relations and the models are left for the final Section \ref{sec:examplesInterpretations}. 
Many proofs are omitted and can be found in the supplementary material.

%
%
%

\section{Asset specialisation and relative performance concerns}
\label{sec2}

We formulate our working market models mimicking the presentations of  \cite{LackerZariphopoulou2017,lackersoret2020many,platonov2020forward}. 

\textbf{The market.} The market environment has one riskless asset and $n$ risky securities which serve as proxies for two distinct asset classes. We assume their prices to be of log-normal type, each driven by two independent Brownian motions. Precisely, the price $(S^i_t)_{t\geqslant 0}$ of stock $i$ traded exclusively by the $i$-th agent solves 
\begin{align}
\label{stock}
\frac{dS_{t}^{i}}{S_{t}^{i}} 
& 
=\breve\mu_{i}dt+\nu_{i}dW_{t}^{i}+\sigma _{i}dB_{t},\quad S^i_0 = s^i_0 > 0,
\end{align}
with constant parameters $\breve\mu_i\in \bR,~\sigma_i\geqs 0$ and $\nu_i\geqs 0$ with $\sigma_i+\nu_i>0$. We refer the reader to \cite{LackerZariphopoulou2017,lackersoret2020many} for an in-depth motivation of the model. The one-dimensional standard Brownian motions $B,W^1,\dots,W^n$ are independent. When $\sigma_i > 0$, the process $B$ induces a correlation between the stocks, and thus we call $B$ the \emph{common noise} and $W^{i}$ an \emph{idiosyncratic noise}. The independent Brownian motions $B,W^1,\dots,W^n$ are defined on a probability space $(\Omega, \bF, \cF, \bP)$ endowed with the natural filtration $\bF=(\cF_t)_{t\geqs 0}$ generated by them and satisfies the usual conditions. 

We recall the case of \textit{single common stock}, where for any $i  = 1,\dots,n ,~ (\breve\mu_i,\sigma_i) = (\breve\mu,\sigma),~ \nu_i = 0,$ for some $\breve\mu \in \bR,~ \sigma >0$ and independent of $i$. 

We contribute to the literature on mean field
games and FPP by also providing an explicitly solvable example. As argued in \cite{lackersoret2020many}, outside linear-quadratic structures such is very rare, and it is one of these rarities we bring here. We work with the very tractable model \eqref{stock} and include common noise, heterogeneous agents, a mean field interaction through the controls in addition to the state processes and FPPs.

\textbf{Agents' wealth.} Each agent $i=1,\ldots,n$ trades using a self-financing strategy, $(\pi _{t}^{i})_{t \geqs 0}$, representing the fraction of wealth invested in the $i$-th stock and consumption policy, $(c^i_t)_{t \geqs 0}$, representing the instantaneous rate of consumption per unit of wealth. The $i$-th agent's wealth dynamics $(X^i_t)_{t\geqs 0}$ is given by  
\begin{equation*}
    dX_{t}^{i} =
    r X_t^idt + \pi_t^iX_t^{i}\Big(\mu_idt+\nu _{i}dW_{t}^{i}+\sigma _{i}dB_{t}\Big)- c^i_t X^i_t dt, \quad \textrm{with}\quad X_0^i=x_0^i > 0, ~\mu_i = \breve \mu_i - r.
\end{equation*}
We interpret $\mu_i$ as an excess return.
A portfolio investment-consumption strategy is deemed \textit{admissible} if it belongs to the admissibility set $\cA^i$,
\begin{align*}
    \cA^i 
    = 
    \Big\{
    (\pi^i,c^i)&: 
    \textrm{$\bF$-progressively measurable $\bR \times (0,\infty)$}- \textrm{valued process }(\pi^i_t,c^i_t)_{t \geqs 0},
    \\
    & \qquad 
    \textrm{such that }
    \bE\Big[\int_0^t (|\pi^i_s|^2 + |c^i_s|^2 ) ds \Big]< \infty,\
     \textrm{for any } t>0
    \Big\}.
\end{align*}
As in \cite{lackersoret2020many} we do not allow a consumption rate of zero. 
It is also clear that for any admissible strategy we have $X^i_t>0$ for all $t\geqs 0$.

\textbf{The agents' interaction and relative performance concerns.} Each manager measures the performance of her strategy taking into account the policies of the others. Each agent engages in a form of social interaction {(in the sense of \cite{EspinosaTouzi2015,BielagkLionnetDosReis2017})} that affects that agent's perception of wealth, all in a multiplicative fashion modelled through the geometric average wealth of all the agents (excluding themselves). The \textit{relative performance wealth process} of manager $i \in \{ 1, \cdots ,n\},$ denoted $\widehat X^i$ is defined to be
\begin{align}
\label{eq:xminusaverage-lessi-consumption}
    \widehat X^i = \frac{X^i}{{\Big(\widetilde{X}^{(-i)}}\Big)^{\theta_i}},
    \qquad\textrm{where}\qquad
    \widetilde{X}^{(-i)} = \Big(\prod_{k \neq i}^nX^k\Big)^{\frac{1}{n-1}}. 
\end{align}
In the same vein we introduce the \emph{relative consumption metric}
\begin{align}
\label{eq:c-hat-i}
 \hat{c}^i = \dfrac{c^i}{\Big(\tilde{c}^{(-i)}\Big)^{\theta_i}}, \quad \text{where} \quad \tilde{c}^{(-i)} = \Big(\prod_{k \neq i}^{n} c^k\Big)^{\frac1{n-1}}.
\end{align}
We emphasise that both relative metrics (hat and tilde) can be equivalently reformulated such that the respective geometric average include the agent themselves. This becomes an equivalent problem that can be reduced to the original one by rescaling the parameters, see \cite[{Remark 3.3}]{LackerZariphopoulou2017},  \cite[Section 2]{lackersoret2020many} or \cite[Section 2]{platonov2020forward}

By an application of It\^o's formula we obtain the dynamics for \emph{the average performance wealth} $Y = \widetilde X^{(-i)}$ as follows,
\begin{align*}
\frac{d Y_t}{ Y_t} 
&
= \left(r 
+ \overline{\mu\pi}^{(-i)}_t 
-  \frac{1}{2}\left(\overline{\Sigma \pi^{2}_t}^{(-i)}
-\big(\overline{\sigma \pi_t}^{(-i)}\big)^{2}
-\frac{1}{n-1} \overline{(\nu \pi_t)^{2}}^{(-i)}\right) - \bar{c}_t^{(-i)}\right)  d t 
\\
&\qquad +\frac{1}{n-1} \sum_{k \neq i}^n \nu_{k} \pi^k_{t} d W_{t}^{k}
+\overline{\sigma \pi_t}^{(-i)} d B_{t}, \quad Y_{0}=\bigg( \prod_{k \neq i}^n x_{0}^{k}\bigg)^{\frac1{n-1}},
\end{align*}
where we define the helpful auxiliary quantities for $t\geqs 0$
\begin{align*}
    &\overline{\mu \pi_t}^{(-i)} = \frac1{n-1}\sum_{k \neq i}^{n} \mu_k \pi^k_t, \quad 
    \overline{(\nu \pi_t)^{2}}^{(-i)} = \frac1{n-1} \sum_{k\neq i}^{n} (\nu_k\pi^k_t)^2,\quad
    \overline{\sigma \pi_t}^{(-i)} = \frac1{n-1} \sum_{k \neq i}^{n} \sigma_k \pi^k_t 
    \\
    &\overline{\Sigma \pi^2_t}^{(-i)} = \frac1{n-1}\sum_{k \neq i}^{n} \Sigma_k (\pi^k_t)^2,\quad
    \bar{c}^{(-i)}_t = \frac1{n-1} \sum_{k \neq i}^{n} c^k_t, \quad  \Sigma_k = \sigma_k^2 + \nu_k^2.
\end{align*} 
Via It\^o's formula one finds the dynamics of \emph{relative performance wealth} $\widehat X^i$ to be
\begin{align}
    \label{eq:widehat-x-consumption}
    \frac{d \widehat X^i_t}{\widehat X^i_t} 
    = \xi_i dt 
    &-  \big(c^i_t - \theta_i\bar{c}_t^{(-i)}\big) dt 
    \\ \nonumber 
    & + \bigg(\nu_i\pi^i_tdW^i_t - \theta_i\Big( \frac1{n-1}\sum_{k \neq i}^{n} \nu_k \pi^k_t dW^k_t\Big)\bigg) + \bigg(\sigma_i \pi^i_t - \theta_i \overline{\sigma \pi_t}^{(-i)}\bigg) dB_t,
\end{align}
where 
\begin{align*}
    \xi_i = r(1-\theta) + \mu_i \pi^i_t 
    &
    - \theta_i \overline{\mu \pi_t}^{(-i)} + \frac{\theta_i}2\overline{\Sigma\pi_t^2}^{(-i)} 
    \\
    &- \frac{\theta_i^2 }{2}\bigg(\big(\overline{\sigma\pi_t}^{(-i)}\big)^2 + \frac{1}{n-1}\overline{(\nu\pi_t)^2}^{(-i)}\bigg) - \theta_i \sigma_i\pi^i_t \overline{\sigma \pi_t}^{(-i)} .
\end{align*}

%
%
%
%
%

%
%
%

\section{The finite player forward optimisation game}
\label{sec:FRPCwithaverage}
\subsection{Forward relative performance}

Each manager $i\in\{1,\dots,n\}$ measures the output of her relative performance metric using a forward relative utility as modelled by an $\cF_t$-progressively measurable random field $Q^i:\Omega \times (0,\infty) \times [0,\infty) \to \bR$ for $i\in\{1,\dots,n\}$. The below criteria follows those proposed in \cite{geng2017passive} (power-type) and later in \cite{platonov2020forward} (exponential-type) for a forward-performance game.

The formulation here is inspired in the first step of the usual strategy of solving Nash equilibria games, namely, the best response of an agent to the actions of all other agents. 
This version of a relative criterion is (implicit and) endogenously parametrised by the policies of all other managers $j\neq i$ over which no assumption on their optimality is made. 
\begin{definition}[Forward relative performance for the manager]
\label{def:ForwardUtilityBestResponse-consumption}
Each manager $i\in\{1,\dots,n\}$ satisfies the following conditions: let for any $j\neq i,~(\pi^j,c^j)\in\cA^j$ be arbitrary but fixed and admissible policies, in other words, the other managers have fixed their investment-consumption admissible strategies. 

For $(\pi^i,c^i) \in \cA^i$ and subjective discount factor $\rho_i \geqs 0$, define the $\bF$-progressively measurable random field $Q^i:\Omega\times (0,\infty) \times [0,\infty) \to \bR$
\begin{align}
\label{eq:structure-Q-n-player}
Q^i(x,t):= e^{-\rho_i t} U^{i}(x,t) + \int_0^t e^{-\rho_i s} V^{i}(\hat c_s^i x,s)ds,
\end{align}
where $\hat c^i$ is given by \eqref{eq:c-hat-i} and 
$U^i,V^i:\Omega\times (0,\infty)\times [0,\infty)\to \bR$ are two other $\bF$-progressively measurable random fields \footnote{The first term of \eqref{eq:structure-Q-n-player} corresponds to the utility that the agent derives at time $t$ from owning the amount of wealth $x$. The second term captures the utility the agent accumulates from time $0$ up to time $t$ from consuming at the rate of $\hat c^i x$. 
Hereinafter, and for simplicity, we call $U$ the utility from wealth and $V$ the utility from consumption.}.

The random field $Q^i$ is a \emph{forward relative performance process for $i$-th manager} if, for all $t\geqs 0$, the following conditions hold:
\begin{itemize}
	\item The mappings $x \mapsto U^{i}(x,t)$ and $x \mapsto V^{i}(x,t)$ are $\bP$-a.s.~strictly increasing and strictly concave;
    \item 
    For any $(\pi^{i},c^{i})\in \cA^i,~Q^i(\widehat X^i_t,t)$ is a (local) supermartingale, where $\widehat X^i$ is the relative performance wealth process given in \eqref{eq:xminusaverage-lessi-consumption};
    \item There exists $(\pi^{i,*},c^{i,*})\in \cA^i$ such that $Q^i(\widehat X^{i,*}_t,t)$ is a (local) martingale, where $\widehat X^{i,*}$ solves \eqref{eq:xminusaverage-lessi-consumption} with the strategies $(\pi^{i,*},c^{i,*})$ being used. The strategy $(\pi^{i,*},c^{i,*})$ is said to be optimal.
\end{itemize}
\end{definition}

In the above definition, we do not make explicit references to the initial conditions $U^i(x,0)$, $V^i(x,0)$ but we assume that admissible ($\cF_0$-measurable) initial data exists such that the above definition is viable. Contrary to the classical expected utility case, the forward performance process is a manager-specific input. Once it is chosen, the supermartingale and martingale properties impose certain conditions on the drift of the process. Under enough regularity, these conditions lead to the forward performance SPDE (see \cite{MusielaZariphopoulou2010,KarouiMrad2013}) which, in our case, reduces to a PDE with stochastic coefficients (see Proposition \ref{prop:BestResponses-n-playerGame-average-n-1-consumption} below).

Since we are working in a log-normal market, it suffices to study smooth relative performance criteria of zero volatility (of the FPP map). Such processes are extensively analysed in \cite{MusielaZariphopoulou2010-space-time-monotone} in the absence of relative performance concerns. There, a concise characterisation of the forward criteria is given along necessary and sufficient  conditions for their existence and uniqueness. In that setting, the zero-volatility forward processes are always time-decreasing processes. We point to the reader that this does not have to be case if relative performance concerns are present (see also \cite{geng2017passive,platonov2020forward}).
We now make a standing assumption regarding the regularity of the FPP map.
\begin{assumption}
    \label{ass:regularity-u-n-player-consumption}
    Assume the partial derivatives $U^i_t(x,t)$, $U^i_{x}(x,t)$, $U^i_{xx}(x,t)$ and $V^i_{x}(x,t)$, $V^i_{xx}(x,t)$ exist for all $t \geqs 0, ~x > 0,~\bP$-a.s. 
    
    The maps, $x \mapsto U^i(x,t)$ and $x \mapsto V^i(x,t)$ are strictly increasing ($U^i_x,V^i_x > 0$) and strictly concave ($U^i_{xx},V^i_{xx} < 0$) for any $t \geqs 0,~ x > 0,~ \bP$-a.s.. Furthermore, $\int_0^t |U^i(x,s)|^2 ds < \infty$, for any $x > 0, ~ t \geqs 0$, $\bP$-a.s. 
\end{assumption}
From Assumption \ref{ass:regularity-u-n-player-consumption}, for $i\in\{1,\dots,n\}$ the It\^o decomposition of the forward map is
\begin{align}
\label{eq:SDEforForwardUtilityField-consumption}
    d Q^i(x,t) = e^{-\rho_i t} U^i_t(x,t) dt - \rho_i  e^{-\rho_i t}U^i(x,t) dt + e^{-\rho_i t}V^i(\hat c_t^i x,t) dt,
    \quad Q^i(x,0) = u^i_0(x).
\end{align}

For posterior use, we recall the notion of Fenchel-Legendre transform applied to a random field $V$ under the above assumption.
\begin{definition}
Let $V:\Omega\times {(0,\infty)} \times [0,\infty)  \to \bR$ be a random field such that $x\mapsto V(x,t)$ is a $\bP$-a.s.~strictly concave function for all $t\geqs 0$. Define the random field $\widetilde V: \Omega\times {(0,\infty)} \times [0,\infty) \to \bR$ as $\widetilde V(x^\prime,t) := \sup_{x > 0} \big \{ V(x,t) - x^\prime x\big\}$ for any $x^\prime > 0$, $t\geqs 0$ and $\omega \in \Omega$.

Then, we call $\widetilde V$ the \emph{Fenchel-Legendre transform} of $V$.
\end{definition}

For a map $V$ satisfying Assumption \ref{ass:regularity-u-n-player-consumption}, its differentiability and strict concavity ensures $\widetilde V$ is always well-defined and can be computed (up to a closed form).

\subsection{Best responses}

Here we derive a PDE with random coefficients and an optimal investment-consumption strategy for a smooth relative performance criteria of zero-volatility of some agent $i$ assuming that all other agents $j\neq i$ have made their investment decisions.
\begin{proposition}[Best responses]
\label{prop:BestResponses-n-playerGame-average-n-1-consumption}
Fix $i\in\{1,\dots,n\}$.
Assume that each manager $j\neq i$ follows $(\pi^j,c^j)\in\cA^j$. Consider the PDE with stochastic coefficients for $(x,t)\in (0,\infty) \times [0,\infty)$ given by
\begin{align}
\label{eq:SPDE-step01-consumption}
\nonumber
U^{i}_t
= & \rho_i U^i + \bigg(\theta_i\overline{ \mu\pi_t}^{(-i)} -r(1-\theta) - \frac{\theta_i \sigma_i \overline{ \sigma\pi_t}^{(-i)} (\mu_i - \theta_i \sigma_i \overline{ \sigma\pi_t}^{(-i)})}{\nu_i^2+\sigma_i^2} - \frac{\theta_i}2 \overline{\Sigma \pi_t^2}^{(-i)} 
\\ 
& \hspace{2cm}
- \frac{\theta_i^2}{2} \big(\big(\overline{\sigma\pi_t}^{(-i)}\big)^2 + \frac1{n-1} \overline{(\nu\pi_t)^2}^{(-i)}\big)\bigg)xU^{i}_x
\\
\nonumber
& +  \frac{(\mu_i - \theta_i \sigma_i \overline{\sigma\pi_t}^{(-i)})^2}{2(\nu_i^2+\sigma_i^2)} \frac{(U^{i}_x)^2}{U^{i}_{xx}}
+  \frac12\bigg(  \Big(\theta_i\overline{\sigma\pi_t}^{(-i)}\Big)^2 \Big( \frac{ \sigma_i^2 }{\nu_i^2+\sigma_i^2} -1 \Big)-\frac{\theta_i^2}{n-1}\overline{(\nu\pi_t)^2}^{(-i)}\bigg) x^2 U^{i}_{xx}
\\
\nonumber
&  + \theta_i  \bar c_t^{(-i)}U_x^i - \widetilde {V}^i (U^i_x,t),
\end{align}
where $\widetilde V^i$ stands for Fenchel-Legendre transform of $V^i$ in variable $x$.  Assume that for admissible initial conditions $U^i(\cdot,0)=u^i_0(\cdot),~ V^i(\cdot,0)=v^i_0(\cdot)$, the PDE has a smooth solution $(U^i,V^i)$, that is not necessarily unique, but satisfy Assumption \ref{ass:regularity-u-n-player-consumption}.

Define the strategy $(\pi^{i,*}, c^{i,*})$
\begin{align}
\label{eq:best-responses-pi-n-player-consumption}
\pi^{i,*}_t 
&=\frac1{\nu_i^2+\sigma_i^2} \Big({\theta_i \sigma_i \overline{\sigma\pi_t }^{(-i)}} 
- \big(\mu_i-\theta_i \sigma_i \overline{\sigma\pi_t }^{(-i)}\big)\frac{U^{i}_x(\widehat X^{i,*}_t ,t)}{U^{i}_{xx}(\widehat X^{i,*}_t ,t)\widehat X^{i,*}_t}\Big), 
\\
\label{eq:best-responses-c-n-player-consumption}
c^{i,*}_t 
&= \dfrac{(V^i_x)^{-1}\Big(U^i_x(\widehat X^{i,*}_t,t) \big(\tilde{c}^{(-i)}_t\big)^{\theta_i},t\Big)\big(\tilde{c}^{(-i)}_t\big)^{\theta_i}}{\widehat X^{i,*}_t},
\end{align}
where $\widehat X^{i,*}$ solves \eqref{eq:widehat-x-consumption} with $(\pi^{i,*},c^{i,*})$ being used.

Then, in the sense of Definition \ref{def:ForwardUtilityBestResponse-consumption}, if $(\pi^{i,*}, c^{i,*})\in\cA^i$ and if $\widehat X^{i,*}$ is well-defined, then $Q^i(x,t)$ is a forward relative performance process for manager $i$ and, moreover, the policy $(\pi^{i,*},c^{i,*})$ is optimal.
\end{proposition}

Let us recall the concept of a CRRA utility map. Following \cite[{Section 5}]{MusielaZariphopoulou2009}, we say a utility map $U$ is of \textit{Constant Relative Risk Aversion} (CRRA) type if the \emph{local risk tolerance function} $r:\Omega\times (0,\infty) \times [0,\infty) \to \bR$, given by the quotient $r(x,t)=-U_x(x,t)/U_{xx}(x,t)$ is linear in space, i.e., $r(x,t) = \delta x,$ for any $\delta > 0$ and $t \geqs 0,~x > 0$. This is the case for the classical power utility function, see Section \ref{sec:FPPinitapowerrisk} next. 

By direct inspection of the expression for $\pi^{i,*}$ in \eqref{eq:best-responses-pi-n-player-consumption}, if the local risk tolerance function satisfies $r^i(x,t)= \delta_i x$, for any $t > 0$ (e.g., a CRRA utility) then the optimal investment strategy is constant throughout time if additionally all other agents also choose a constant investment strategy. 

\begin{corollary}[Constant investment strategies under CRRA]
Assume $U^i$ and $V^i$ to be of CRRA type, i.e., the \emph{local risk tolerance function} is linear in space uniformly. Assume further that all agents $j\neq i$ invest according to constant investment strategies $\pi^j\in\bR$. Then, $\pi^{i,*}$ is constant.

\end{corollary}
At this point it is not clear that if the agents $j\neq i$ use deterministic continuous consumption strategies $c^j$ then the optimal consumption strategy $c^{i,*}$ in \eqref{eq:best-responses-c-n-player-consumption} is also so. 
We next prove the ``best responses'' result, Proposition \ref{prop:BestResponses-n-playerGame-average-n-1-consumption}.
\begin{proof}[Proof of Proposition \ref{prop:BestResponses-n-playerGame-average-n-1-consumption}]
From \eqref{eq:xminusaverage-lessi-consumption} we have access to the dynamics of $d(X^i_t\big(\overline{X}^{(-i)}_t\big)^{-\theta_i})$ and $d\widehat X^i$. Using \eqref{eq:SDEforForwardUtilityField-consumption}, we apply the It\^o formula to $Q^i(\widehat X^i_t ,t) = Q^i\big(X^i_t \big(\overline{X}^{(-i)}_t\big)^{-\theta_i} ,t\big)$, and obtain using the notation set before \eqref{eq:widehat-x-consumption}
\begin{align}
\label{eq:generalSDEafterItoWentzell-consumption}
\nonumber
d Q^i(\widehat X^i_t ,t) 
&= e^{-\rho_i t} \Big[ U^{i}_t(\widehat X^i_t ,t)dt  - \rho_i U^i(\widehat X^i_t, t) dt + U^{i}_x(\widehat X^i_t ,t) d \widehat X^i_t + \frac12 U^{i}_{xx}(\widehat X^i_t ,t) d \langle \widehat X^i_t \rangle \\
\nonumber
& \qquad + V^{i}(\hat c^i_t\widehat X_t^i, t) dt \Big] 
\\ 
\nonumber
&
=e^{-\rho_i t} \bigg[ U^{i}_t(\widehat X^i_t ,t)dt - \rho_i U^i(\widehat X^i_t, t) dt
+ U^{i}_x(\widehat X^i_t ,t)\bigg(\mu_i \pi^i_t - \theta_i \overline{\mu \pi_t}^{(-i)} + r(1-\theta) \\
\nonumber
& \qquad + \frac{\theta_i}2\overline{\Sigma\pi_t^2}^{(-i)} + \frac{\theta_i^2}{2}\Big(\big(\overline{\sigma\pi_t}^{(-i)}\big)^2 + \frac{1}{n-1}\overline{(\nu\pi_t)^2}^{(-i)}\Big)- \theta_i \sigma_i\pi^i_t \overline{\sigma \pi_t}^{(-i)}\bigg)\overline{X}^i
_t dt 
\\
& \qquad
+ U^{i}_x(\widehat X^i_t ,t)\big(\sigma_{i}\pi_t^i - \theta_i\overline{\sigma\pi_t}^{(-i)} \big)\widehat{X}^i
_tdB_t 
\\
\nonumber
& \qquad + U^{i}_x(\widehat X^i_t ,t) \widehat{X}^i_t \bigg(\nu_i\pi^i_tdW^i_t - \theta_i\Big( \frac1{n-1}\sum_{k \neq i}^{n} \nu_k \pi^k_t dW^k_t\Big)\bigg) 
\\
\nonumber
&\qquad + \frac12 U^{i}_{xx}(\widehat X^i_t ,t)\bigg( (\nu_i\pi^i_t )^2 
+\frac{\theta_i^2}{n-1}\overline{(\nu\pi_t )^2}^{(-i)}
+ \Big(\sigma _{i}\pi_t^i- \theta_i\overline{\sigma\pi_t}^{(-i)}\Big)^2
\bigg)\big(\widehat{X}^i_t\big)^2dt\\
\nonumber
& 
\qquad - U^{i}_x(\widehat X^i_t ,t)\big(c^i_t - \theta_i\bar{c}^{(-i)}_t \big)\widehat X_t^i dt + V^{i}(\hat c^i_t\widehat X_t^i, t) dt \bigg],
\end{align}
with $Q^i(\widehat X^i_0 ,0)=Q^i\big(x^i_0(\overline x_0^{(-i)})^{-\theta_i},0\big)=u^i_0\big(x^i_0(\overline x_0^{(-i)})^{-\theta_i}\big)$ and that the $B,W^j$ are all i.i.d.

By Definition \ref{def:ForwardUtilityBestResponse-consumption}, the process $Q^i(\widehat X^i_t,t)$ becomes a martingale at the optimal $(\pi^{i,*},c^{i,*})$, hence, direct computations using first order conditions ($\partial_{\pi^i} \textrm{``drift''}= \partial_{c^i} \textrm{``drift''}=0$) yield
\begin{align}
\label{eq:optimstrategy-consumption}
\begin{cases}
 \pi^i_t 
&=
\dfrac1{\nu_i^2+\sigma_i^2} \Big({\theta_i \sigma_i \overline{\sigma\pi_t}^{(-i)}} - \big(\mu_i - \theta_i \sigma_i \overline{\sigma \pi_t}^{(-i)}\big)\dfrac{U^i_x(\widehat X^i_t ,t)}{U^i_{xx}(\widehat X^i_t ,t)\widehat X^i_t}\Big), 
\\
c^i_t &= \dfrac{(V^i_x)^{-1}\Big(U^i_x(\widehat X^i_t,t) \big(\tilde{c}^{(-i)}_t\big)^{\theta_i},t\Big)\big(\tilde{c}^{(-i)}_t\big)^{\theta_i}}{\widehat X_t^i}.
\end{cases}
\end{align}
Injecting $\pi^i_t $ in the drift term of \eqref{eq:generalSDEafterItoWentzell-consumption} and simplifying we arrive at the consistency condition \eqref{eq:SPDE-step01-consumption}, we do not carry out this step explicitly, nonetheless, using that the pair $(U^{i},V^{i})$ solves \eqref{eq:SPDE-step01-consumption}, equation \eqref{eq:generalSDEafterItoWentzell-consumption} simplifies to
\begin{align}
\label{eq:SDEafterItoWentzell and choice of Ut-average-n-1-consumption}
\nonumber
    d Q^i(\widehat X^i_t ,t) 
    &
    = e^{-\rho_i t} \bigg[ U^{i}_x(\widehat X^i_t ,t) \widehat X^i_t
    \Big(\nu _{i}\pi_t^idW_{t}^{i}-\theta_i\big(\frac1{n-1}\sum_{k\neq i}^n \pi_t^k\nu_{k}dW_{t}^{k}\big)\Big)
    \\ \nonumber
    &\quad 
    + U^{i}_x(\widehat X^i_t ,t)\Big(\sigma_{i}\pi_t^i - \theta_i\overline{\sigma\pi_t}^{(-i)} \Big)\widehat X^i_t dB_t\\ 
    &
    \quad +\frac12 U^{i}_{xx}(\widehat X^i_t ,t)\frac1{\nu_i^2+\sigma_i^2} \bigg| {\pi^i_t}(\nu_i^2+\sigma_i^2) -\Big( \theta_i \sigma_i \overline{ \sigma\pi_t}^{(-i)} - \mu_i\frac{U^i_x(\widehat X^i_t ,t)}{U^i_{xx}(\widehat X^i_t ,t)\widehat X^i_t} \Big)\bigg|^2(\widehat X^i_t)^2dt 
    \\
    \nonumber 
    &
    \quad+ V^{i}(\hat c^i_t\widehat X_t^i, t) dt - U^{i}_x(\widehat X^i_t ,t)c^i_t \widehat X_t^i dt -  \widetilde V^i(U^i_x(\widehat X_t^i,t),t)dt \bigg].
\end{align}
The concavity assumption on $U^{i}(x,t)$ implies that the third term of the expression above is non-positive and vanishes when \eqref{eq:optimstrategy-consumption} holds. At the same time by the definition of $\widetilde V^i$, the Fenchel-Legendre transform of $V^i$, we notice that 
$$
\widetilde V^i\big(U^i_x(\widehat X^i_t,t),t\big) 
= V^i\big(\hat c^{i,*}_t\widehat X_t^i,t\big)-U^i_x\big(\widehat X_t^i,t\big)c^{i,*}_t \widehat X^i_t,
$$
yields the non-positivity and extremality of the last line of \eqref{eq:SDEafterItoWentzell and choice of Ut-average-n-1-consumption} when $c^i = c^{i,*}$.  We can conclude that, if $(\pi^{i,*}_t,c^{i,*}_t)= (\pi^i_t,c^i_t)\in \cA^i$ and the associated process $\widehat X^{i,*}$ is well-defined (solution to \eqref{eq:widehat-x-consumption} with $(\pi^{i,*},c^{i,*})$), the process $U^i(\widehat X^{i,*}_t,t)$ is a local-martingale, otherwise it is a local supermartingale. The result concludes.
\end{proof}
In contrast with \cite{lackersoret2020many} we do not know the explicit form of the utility as it is part of the solution. Hence, we exploit convex duality properties to prove the extremality argument in combination with the HJB-type methodology. In \cite{geng2017passive} or \cite{platonov2020forward} the authors only explored the investment problem and this issue does not appear. 

\subsection{Forward performance with initial CRRA power preference}
\label{sec:FPPinitapowerrisk}

For $i\in\{1,\dots,n\}$, the dynamics of $Q^i$ is given by \eqref{eq:SDEforForwardUtilityField-consumption}.
Then, the solution to the PDE given in Equation \eqref{eq:SPDE-step01-consumption} has the following form
\begin{align}
\label{eq:Solution Postulation-consumption} 
U^i(x,t) & = \begin{cases}
    \frac1{1-\frac1{\delta_i}} x^{1-\frac1{\delta_i}}f_i(t), & \delta_i \neq 1,~ \delta_i >0,\\
    \log(x)f_i(t) + h_i(t), & \delta_i = 1,
    \end{cases}
\end{align}
where $x > 0$ and $(f_i(t))_{t\geqs 0},(h_i(t))_{t\geqs 0}$ are the maps independent of $x$ satisfying $f_i(0)=1$ and $h_i(0) = 0$, are sufficiently integrable and $t\mapsto f_i(t),~ t\mapsto h_i(t)$ are differentiable.
\begin{remark}
\label{rem:skipping delta=1 case}
For the sake of simplicity, we omit the presentation for the logarithmic utility and emphasise that the optimal policy for this case coincides with the general power case with $\delta_i = 1$ plugged inside.
\end{remark}
The structure of the  consistency condition \eqref{eq:SPDE-step01-consumption} implies $V^i$ (see \cite[Proposition 4.5]{elkaroui2018consistent}) as
\begin{align}
\label{eq:Solution Postulation-consumption-V} 
    V^i(x,t) & = \frac1{\epsilon_i}\frac1{1-\frac1{\delta_i}} x^{1-\frac1{\delta_i}} g_i(t),
\end{align}
where $x > 0$ and $(g_i(t))_{t\geqs 0}$ is a map independent of $x$ satisfying $g_i(0)=1$ and sufficiently integrable. We refer to the constants $\delta_i>0,~ \delta_i \neq 1,$ and $\theta_i\in[0,1]$ as \emph{personal risk tolerance} and \emph{competition weight}, respectively, whereas $\epsilon_i > 0$ captures the \emph{relative importance the agent assigns to the wealth compared to consumption}. Finally, $\rho_i$ stands for \emph{personal discount rate} of the agent.

In this case of initial power preferences we have that the \emph{local risk tolerance function} for $U^i$ satisfies $r^i=-U^i_x/U^i_{xx}=\delta_ix$, and hence
\begin{align*}
    \pi^{i,*}_t &=\frac1{\nu_i^2+\sigma_i^2} \bigg({\theta_i \sigma_i \overline{ \sigma\pi_t}^{(-i)}} + \left(\mu_i - \theta_i \sigma_i \overline{\sigma\pi_t }^{(-i)}\right)\delta_i\bigg).
\end{align*}
We now proceed to find this case's optimal consumption map $c^i$ and utility fields. Injecting the expressions above for $U^i(x,t),~V^i(x,t)$ in \eqref{eq:SPDE-step01-consumption} together with the optimal consumption policy given by \eqref{eq:optimstrategy-consumption} yields an ODE system for $(f_i, g_i, c^i)$. We have
\begin{align}
\label{eq:system-equations-n-player-consumption}
\begin{cases}
    c^i_t = \big(\tilde c_t^{(-i)}\big)^{\theta_i(1-\delta_i)}\epsilon_i^{-\delta_i} \bigg( \dfrac{g_i(t)}{f_i(t)}\bigg)^{\delta_i}, \\
    f^\prime_i(t) + \Big(\eta_i + \big(1-\frac1{\delta_i}\big)\theta_i \bar{c}_t^{(-i)}\Big) f_i(t) + \frac{\epsilon_i^{-\delta_i}}{\delta_i} \big(\tilde c_t^{(-i)}\big)^{\theta_i(1-\delta_i)}  f_i(t)
\bigg(\dfrac{g_i(t)}{f_i(t)}\bigg)^{\delta_i} = 0~,
\end{cases}
\end{align}
where
\begin{align}
    \label{eq:lambda-consumption}
    \eta^i_t &=  \Big(1-\frac1{\delta_i}\Big) \bigg(\frac{\delta\big(\mu_i-\theta_i\sigma_i\overline{\sigma\pi_t}^{(-i)}(1-\frac{1}{\delta_i})\big)^2}{2(\nu_i^2+\sigma^2_i)}- r(1-\theta_i)-\theta_i\overline{\mu\pi_t}^{(-i)} + \frac{\theta_i}2 \overline{\Sigma \pi_t^2}^{(-i)}
    \\ \nonumber
    &\qquad\qquad + \frac{\theta_i^2}{2} \Big(\big(\overline{\sigma\pi_t}^{(-i)}\big)^2 + \frac1{n-1} \overline{(\nu\pi_t)^2}^{(-i)}\Big) \big(1-\frac1{\delta_i}\big) \bigg) - \rho_i.
\end{align}
We have two equations for three unknowns, now we need one further assumption for the nature of $f_i$ and $g_i$ in order to solve the system.
\begin{assumption}
    \label{ass:form-of-g-n-player-example}
    Let $g_i(t) = f_i(t)^{1 - \kappa_i}$, where $\kappa_i \in \bR$.
\end{assumption}
\begin{remark}
    Here, as opposed to \cite{lackersoret2020many} or \cite{karatzas1997explicit} where one is naturally led to $\kappa_i=1$, we find a non-trivial time-dependence structure of the consumption utility. We highlight that the FPP approach naturally supports the time dynamics of the utilities in contrast to the HJB approach that does not additionally weight the integrand with a function of time.
\end{remark}
We have two distinct cases.

\emph{Case 1: Let $\kappa_i \neq 0$.} 
    We obtain and solve a classical ODE of \textit{Bernoulli equation} type (we omit the dependence in $t$ for simplification)
\begin{align*}
    f^\prime_i + a_i(t)f_i + b_i(t)f_i^{1 - \kappa_i\delta_i} = 0
    \ \textrm{ for }\ 
        a_i(t) = \eta_t^i - \theta_i \bar{c}^{(-i)}_t 
        \ \textrm{ and }\
        b_i(t) = \frac{\epsilon_i^{-\delta_i}}{\delta_i} \big(\tilde c_t^{(-i)}\big)^{\theta_i(1-\delta_i)}.
\end{align*}
Substituting $k_i(t) := f_i(t)^{-\kappa_i\delta_i}$ we find the ODE for $k_i(t)$, 
\begin{align*}
    \frac1{\kappa_i\delta_i}k^\prime_i - a_i(t)k_i - b_i(t)= 0.
\end{align*}
Hence, solving the ODE we have 
\begin{align*}
k_i(t)=e^{\kappa_i\delta_i\int_0^t a_i(s) ds} + \kappa_i\delta_i \int_0^te^{\kappa_i\delta_i\int_s^t a(r)dr} b_i(s) ds,
\end{align*}
and $f_i(t) = k_i(t)^{\frac1{\kappa_i\delta_i}}$, $g_i(t) = 
f_i(t)^{1 - \kappa_i} =
k_i(t)^{\frac{1-\kappa_i}{\kappa_i\delta_i}}$. Plugging the $f_i,g_i$ expressions in the first equation of \eqref{eq:system-equations-n-player-consumption} along with the general form of the consumption strategy we can extract the form of optimal consumption map.

\emph{Case 2. Let $\kappa_i = 0$.} In this case, the optimal consumption policy of agent $i$ simplifies considerable. More precisely, from \eqref{eq:system-equations-n-player-consumption} we obtain
\begin{align*}
    c^i_t = \big(\tilde c_t^{(-i)}\big)^{\theta_i(1-\delta_i)} \epsilon_i^{\delta_i}.
\end{align*}

\subsection{The Forward Nash equilibrium}

In view of the \emph{best responses} of Proposition \ref{prop:BestResponses-n-playerGame-average-n-1-consumption} we now solve for the \emph{simultaneous best responses} as to establish the existence of a Nash equilibrium.
\begin{definition}[Forward Nash equilibrium]
\label{def:ForwardNashEquilibrium-consumption}
Let for any $i \in \{1,\dots, n\},~(\pi^{i,*},c^{i,*})\in\cA^i$ and $(\pi^{i,*},c^{i,*})$ is the optimal strategy in the sense of Proposition \ref{prop:BestResponses-n-playerGame-average-n-1-consumption}.
Let $\big(\pi^i,c^i\big) \in \cA^i$. Let  $Q^i$ be the $\bF$-progressively measurable random field satisfying 
$$Q^i(x,t):= e^{-\rho_i t}U^{i}(x,t) + \int_0^t e^{-\rho_i s}V^{i}(\hat c_s^i x,s)ds,$$ where $\hat c^i$ is given by \eqref{eq:c-hat-i} with $c^j = c^{j,*},~j\neq i$ and 
$U^i,V^i:\Omega\times (0,\infty) \times [0,\infty)\to \bR$ are two other $\bF$-progressively measurable random fields.

A forward Nash equilibrium consists of $n$-triples of $\cF_t$-adapted maps $(Q^i,\pi^{i,*},c^{i,*})$ with $i=1,\dots,n$ such that for any $t\geqs 0$ the following conditions hold.
\begin{itemize}
    \item[i)] The mappings $x \mapsto U^{i}(x,t)$ and $x \mapsto V^{i}(x,t)$ are $\bP$-a.s.~strictly increasing and strictly concave;
    
    \item[ii)] 
	Let managers $j\neq i$ acting according to $(\pi^{j,*},c^{j,*})$, manager $i$ acts according to $(\pi^i,c^i) \in \cA^i$, and take $\widehat X^{i}$ as the associated \emph{relative performance wealth process} \eqref{eq:widehat-x-consumption} and \emph{relative consumption metric} $\hat c^{i}$  given by \eqref{eq:c-hat-i}. 
	Then $Q^i(\widehat X^{i}_t,t)$ is a (local) supermartingale.
	\item[iii)] 
	Let \emph{all} managers acting according to $(\pi^{j,*},c^{j,*})$, and take $\widehat X^{i,*}$ as the associated \emph{relative performance wealth process} \eqref{eq:widehat-x-consumption} and \emph{relative consumption metric} $\hat c^{i,*}$  given by \eqref{eq:c-hat-i}. 
	Then $Q^i(\widehat X^{i,*}_t,t)$ is a (local) martingale.
\end{itemize}

\end{definition}

\subsection*{Equilibrium with FPPs of separable power factor form}

In order to obtain explicit results we focus on the \textit{separable power factor form} case of \eqref{eq:Solution Postulation-consumption}-\eqref{eq:Solution Postulation-consumption-V} for which $U^i_x/U^i_{xx}=V^i_x/V^i_{xx}=-\delta_i x$. Notably, at the level at which we have formulated our problem we recover the results of \cite[{Theorem 3}]{lackersoret2020many} for which one has $U^i_x/U^i_{xx}=-\delta_ix$, for all $t$ and $x > 0$ (note their Remark 5).
\begin{assumption}
\label{ass:form-of-g-n-player}
    Set $g_i(t) := f_i(t)^{1-\kappa}$ in \eqref{eq:Solution Postulation-consumption}-\eqref{eq:Solution Postulation-consumption-V} for some $\kappa\in \bR$ for any $i = 1,\dots,n$.
\end{assumption}

One can see that the PDE \eqref{eq:SPDE-step01-consumption} depends on the two unknown functions $U$ and $V$, and thus makes for an under-determined system.
Hence, Assumption \ref{ass:form-of-g-n-player} is essential to ensure that the equation is solvable for the unique Nash equilibrium. 
At the same time we have to assume $\kappa$ to be constant across the whole population in order to get a solution for the Nash equilibria.
\begin{theorem}
\label{theo:nPlayerForwardNashGame-consumption}
Let the conditions of Proposition \ref{prop:BestResponses-n-playerGame-average-n-1-consumption} hold for all agents $i\in\{1,\dots,n\}$. Assume furthermore that agents have separable power factor form forward maps $U^i,V^i$ with initial conditions.
\begin{align*}
U^i(x,0) 
= 
\epsilon_iV^i(x,0) 
= 
\frac1{1-\frac1{\delta_i}}x^{1-\frac1{\delta_i}},\quad
\epsilon_i>0,\ \delta_i>0,\ \delta_i \neq 1,
\end{align*}
as in \eqref{eq:Solution Postulation-consumption}-\eqref{eq:Solution Postulation-consumption-V}. Define the quantities 
\begin{align*}
&\varphi^\sigma_{n} =
\frac1n\sum_{k=1}^n  \delta_k \frac{\sigma_k \mu_k }{{\nu_k^2+\sigma_k^2 \big(1+\frac{\theta_k(1-\delta_k)}{n-1}\big) }}
,\qquad
\psi^\sigma_n = 
\frac1{n-1}\sum_{k = 1}^n  \theta_k(1-\delta_k) \frac{\sigma^2_k}{{\nu_i^2+\sigma_i^2 \big(1+\frac{\theta_k(1-\delta_k)}{n-1}\big) }},
\\
&\lambda_i = \epsilon_i^{-\frac{\delta_i}{1 + \frac{\theta_i}{n-1}(1-\delta_i)}} \Big(\widetilde{\epsilon^{\delta}}\Big)^{\frac{\theta_i(1-\delta_i)}{{(\overline{\theta(1-\delta)}-1) \big(1 + \frac{\theta_i}{n-1}(1-\delta_i)\big)}}}, 
\quad
\beta_i 
= \frac1{1 + \frac{\theta_i}{n-1}(1-\delta_i)}\bigg(\frac{\theta_i(1-\delta_i)}{\overline{\theta(1-\delta)}-1} \overline{\eta\delta}-{\eta_i\delta_i}\bigg),
\\
&\widetilde{\epsilon^{\delta}} 
= \Bigg(\prod_{k = 1}^n \epsilon_k^{\frac{\delta_k}{1 + \frac{\theta_k}{n-1}(1-\delta_k)}}\Bigg)^{\frac1n},
\quad 
\overline{\eta\delta} 
= \frac1n \sum_{k = 1}^n \frac{\eta_k \delta_k}{1 + \frac{\theta_k}{n-1}(1-\delta_k)}, 
\\
&
\textrm{and}\quad 
\overline{\theta(1-\delta)} = \frac1{n-1}\sum_{k = 1}^n \frac{\theta_k(1-\delta_k)}{1 + \frac{\theta_k}{n-1}(1-\delta_k)}.
\end{align*}
Let as well, for $i=1,\dots,n$, the map $\eta^i_t$ be defined as  
\begin{align}
    \label{eq:rho-n-player}
    \eta^i_t = \eta_i&:= \Big(1-\frac1{\delta_i}\Big) \bigg(\frac{\delta\big(\mu_i-\theta_i\sigma_i\overline{\sigma\pi}^{(-i)}(1-\frac{1}{\delta_i})\big)^2}{2(\nu_i^2+\sigma^2_i)}-r(1-\theta_i)-\theta_i\overline{\mu\pi}^{(-i)} + \frac{\theta_i}2 \overline{\Sigma \pi^2}^{(-i)}\\
    \nonumber
    &\qquad\qquad + \frac{\theta_i^2}{2} \Big(\big(\overline{\sigma\pi}^{(-i)}\big)^2 + \frac1{n-1} \overline{(\nu\pi)^2}^{(-i)}\Big) \big(1-\frac1{\delta_i}\big) \bigg) - \rho_i,
\end{align}
where the explicit expressions for $\overline {\sigma\pi}^{(-i)},~\overline {\mu\pi}^{(-i)},~\overline {(\nu\pi)^2}^{(-i)}$ and $\overline {\Sigma \pi^2}^{(-i)}$ are found in the supplementary material in \eqref{eq:pisigma-nplayers}, \eqref{eq:pimu-nplayers}, \eqref{eq:pinu-nplayers}, \eqref{eq:piSigma-nplayers} respectively.

If $\psi^\sigma_{n}\neq 1$ and Assumption \ref{ass:form-of-g-n-player} holds ($\kappa_i=\kappa$ for all $i$) then a unique optimal candidate strategy exists with the optimal $\pi^{i,*}$ and $c^{i,*}$ given by
\begin{align}
\label{eq:Nash-pi-n-player-consumption}
    \pi^{i,*}
    &=
    \frac1{{\nu_i^2+\sigma_i^2 \big(1+\frac{\theta_i(1-\delta_i)}{n-1}\big) }}
    \Big(\theta_i \sigma_i(1-\delta_i) \Big(1+\frac1{n-1}\Big) \frac{\varphi^\sigma_n}{1-\psi^\sigma_n}  +  \mu_i \delta_i \Big),
    \\
\label{eq:Nash-c-n-player-consumption}
    c^{i,*}_t &=   
    \begin{cases}
        \bigg(\frac{1}{\beta_i} + \Big(\frac1{\lambda_i} - \frac1{\beta_i}\Big)e^{\kappa\beta_it} \bigg)^{-1}, & \beta_i \neq 0,
        \\
        \big(-\kappa t + \frac1{\lambda_i}\big)^{-1}, & \beta_i = 0,
    \end{cases}
\end{align}
and $\eta^i_t$ in \eqref{eq:rho-n-player} is constant in time. Moreover, {optimal candidate strategy} is a \emph{forward Nash equilibrium strategy} if either $\kappa \leqs 0$ or $\kappa > 0, ~\beta_i > \lambda_i$.

The forward Nash equilibria is given by the $n$-triples $\{(Q^{i,*},\pi^{i,*},c^{i,*})\}_{i=1,\dots,n}$ where the $Q^{i,*}(x,t) = e^{-\rho_i t} U^{i,*}(x,t) + \int_0^t e^{-\rho_i s}V^{i,*}(\hat{c}^i_sx,s) ds$, with $U^{i,*},V^{i,*}$ the solution of \eqref{eq:SPDE-step01-consumption} (of the form \eqref{eq:Solution Postulation-consumption}-\eqref{eq:Solution Postulation-consumption-V}) with the optimal strategies $\pi^{\cdot,*},c^{\cdot,*}$ plugged-in. The maps $f_i,g_i$ can be determined up to a closed formula given below in \eqref{eq:f-final-form-n-player}.

If $\psi^\sigma_n=1$ {or $\kappa > 0, ~\beta_i \leqslant \lambda_i$} then there exists no Nash equilibrium.
\end{theorem}

\begin{proof}
The complete proof can be found in Section \ref{supsec:proof-N-game} of the supplementary material.
\end{proof}
The parameter $\kappa$ is interpreted as the \textit{market-risk relative consumption preference} and in Section \ref{sec:examplesInterpretations} below we discuss at length its economic features and comparative interpretation. 
The investment strategy  \eqref{eq:Nash-pi-n-player-consumption} is a classical expression matching that of \cite{LackerZariphopoulou2017,lackersoret2020many,platonov2020forward}. The expression for the consumption strategy \eqref{eq:Nash-c-n-player-consumption} has similarities to the results found in \cite[Equation (5)]{lackersoret2020many} with the crucial difference of $\kappa$. When  $\kappa = 1$ then the results coincide with \cite{lackersoret2020many}.



For some combination of parameters $\lambda_i, \beta_i$ and $\kappa,~(c^{i,*}_t)_{t\geqs 0}$ is not admissible due to the finite-time blow-up.  We comment further on the admissibility of $c^*$ in Section \ref{sec:examplesInterpretations}.


\begin{remark}[Open question] 
\label{rem:open-question-after-nash-n-player}
To obtain a Nash equilibria result, we are restricted to $\kappa$ being constant across the agents, and the same will happen for the MFG case in the next Section. The constraint is technical and stems from the difficulty in expressing the consistency equation for the averages of all $c^i$ (i.e., \eqref{eq:consistency-geom-aver-c-n-player} is not longer reachable) and without such an equation the system cannot be solved. 
At the same time, for a ``best response'' all $\kappa_i$'s may be distinct. 
The open question is then if one allows for different $\kappa_i$'s, how to carry out the aggregation procedure as to express the consistency equation for the averages of all $c^i$? 

In \cite{BielagkLionnetDosReis2017} 
An example of how this problem was overcome can be found . There, aggregation was achieved through the so-called \textit{weighted-dilated inf-convolution}, and the authors found the dynamic equation of the \textit{Representative agent} of their economy, which allowed them to solve their Nash-equilibrium problem. The \textit{representative agent} methodology remains unexplored in the context of the FPP. For instance, it is unclear how inf-convolution can be carried out in this framework as the FPP is not known in advance. It is endogenously found as part of the problem's solution in opposition to the classic utility game-type problems.  

\end{remark}

Lastly, if one sets $\kappa = 1$ then our FPP solution can be rewritten in respect to a time horizon $T$ to recover exactly the results of \cite{lackersoret2020many} (and \cite{LackerZariphopoulou2017} if $V^i=0$) and those of \cite{geng2017passive} for two players and no consumption. Such calculations are straightforward hence omitted. We close the section with a corollary for the single stock case.
\begin{corollary}
\label{cor:single-stock-n-player}
Let $\mu_i = \mu,~\sigma_i = \sigma,~\nu_i  = 0$, for all $i = 1,\dots, n$ and $\mu,\sigma > 0$. Then $\lambda_i,~\beta_i,~\theta_{\text{crit}},~ \overline{\rho\delta}, ~ \overline{\delta}$ are expressed as
\begin{align*}
    \lambda_i &= \epsilon_i^{-\frac{\delta_i}{1 + \frac{\theta_i}{n-1}(1-\delta_i)}} \Big(\widetilde{\epsilon^{\delta}}\Big)^{\frac{\theta_i(1-\delta_i)}{{(\overline{\theta(1-\delta)}-1) (1 + \frac{\theta_i}{n-1}(1-\delta_i))}}},
    \quad 
    \theta_{\textrm{crit}} = \dfrac{1-\frac1{n-1}\sum_{k=1}^n\frac{\theta_k(1-\delta_k)}{1+\frac{\theta_k}{n-1}(1-\delta_k)}}{\overline{\delta}},
    \\
    \overline{\delta} &= \frac1{n-1}\sum_{k=1}^n \frac{\delta_k}{1+\frac{\theta_k}{n-1}(1-\delta_k)},\quad\overline{\rho \delta} = \frac1{n}\sum_{k=1}^n \frac{\rho_k\delta_k}{1+\frac{\theta_k}{n-1}(1-\delta_k)}
    \\
    \beta_i &= \frac{\mu^2}{2\sigma^2\big(1+\frac{\theta_i}{n-1}(1-\delta_i)\big)}\Big(1-\frac{\delta_i}{1 + \frac{\theta_i}{n-1}(1-\delta_i)}\Big)\Big(1 -\frac{\theta_i}{\theta_{\textrm{crit}}\big(1+\frac{\theta_i}{n-1}\big)}\Big)\Big(\delta_i + \frac{\theta_i}{\theta_{\textrm{crit}}}(1-\delta_i) \Big)\\
    &\qquad + \frac{r}{ \big(1+\frac{\theta_i}{n-1}(1-\delta_i)\big)} (1-\delta_i) \left(1-\frac{\theta_i}{\theta_{\text{crit}}}\right) 
    + \frac1{ \big(1+\frac{\theta_i}{n-1}(1-\delta_i)\big)} \frac{\theta_i}{\theta_{\text{crit}}} \frac{\overline{\rho \delta}}{\overline{\delta}}(1-\delta_i) + \rho_i \delta_i.
\end{align*}

Then, if $\theta_{\text{crit}}\neq 1$ the optimal candidate strategy exists with the optimal $(\pi^{i,*},c^{i,*})$ 
\begin{align*}
    \pi^{i,*} &= \frac{\mu}{\sigma^2 \big(1+\frac{\theta_i}{n-1}(1-\delta_i)\big)} \Big(\delta_i + \frac{\theta_i}{\theta_{\textrm{crit}}}(1-\delta_i) \Big),
    \\
    c^{i,*}_t &=
    \begin{cases}
        \bigg(\frac{1}{\beta_i} + \Big(\frac1{\lambda_i} - \frac1{\beta_i}\Big)e^{\kappa\beta_i t} \bigg)^{-1}, & \beta_i \neq 0,\\
        \big(-\kappa t + \frac1{\lambda_i}\big)^{-1}, & \beta_i = 0.
    \end{cases}
\end{align*}
Furthermore, the optimal candidate strategy is a Nash equilibrium, if either $\kappa < 0$, or $\kappa >0,~ \lambda_i <  \beta_i$.
\end{corollary}
%
%
%
%
%

%
%
%
\section{The mean field forward optimisation game} \label{sec:MFG-forward-consumption}

Inspecting Theorem \ref{theo:nPlayerForwardNashGame-consumption} one sees that the optimal strategies and FPP of an agent depend on that agent's specific parameters (model parameters, initial wealth, risk tolerance and performance concern) and certain averages of the parameters of all agents. 
This motivates an MFG approach to the game. In this section, and inspired by \cite{platonov2020forward}, we formalise the concept of CRRA forward mean field Nash game. We use the concept of \emph{type distributions} introduced in \cite{LackerZariphopoulou2017,lackersoret2020many,platonov2020forward}. 

We focus on forward maps that at time $t=0$ are of power type, 
\begin{align*}
U^i(x,m,0) & = \epsilon_i V^i(x,m,0) = \begin{cases}
    \frac{1}{1-\frac1{\delta_i}}\big(x m^{-\theta_i}\big)^{1-\frac1{\delta_i}}, & \delta > 0,~\delta_i \neq 1,\\
    \log(xm^{-\theta_i}), & \delta_i=1,
\end{cases}
\end{align*}
where $x>0,m>0$ denote the wealth of agent and the average wealth of other agents respectively (for $U^i$) or the consumption of the agent and the average consumption of the other agents respectively (for $V^i$). 
We refer to $\delta_i>0$ and $\theta_i\in[0,1]$ as \emph{personal risk tolerance}, \emph{competition weight} parameters, respectively, whereas $\epsilon_i > 0$ captures the \emph{relative importance the agent assigns to the wealth compared to consumption}. 
Finally, for the admissible set of strategy pairs $(\pi^{i,*},c^{i,*}), ~ i = 1, \dots, n$, we recall the \textit{Forward relative performance process} $Q^i(x,t):= {e^{-\rho_i t}} U^{i}(x,t) + \int_0^t {e^{-\rho_i s}} V^{i}(\hat c_s^i x,s)ds$, where $c_s^i$ is given by \eqref{eq:c-hat-i}. We refer to $\rho_i$ as \emph{personal discount rate}.
Here, as in the previous section (see Remark \ref{rem:skipping delta=1 case}), we skip the logarithmic case.

For the $n$-agent game, we define for each agent $i=1,\dots,n$ the \emph{type vector}
\begin{equation*}
\label{eq:type-vector}
\zeta_i := (x^i_0,\delta_{i},\theta_{i},\epsilon_i,\rho_i, \breve\mu_{i},\nu_{i},\sigma_{i}),
\end{equation*}
which characterises uniquely each agent $i$.
These \emph{type vectors} induce an empirical measure, called the \emph{type distribution}, which is probability measure on the \emph{type space}
\begin{align}
\label{def:exp-type-space}
\cZ^e  := (0,\infty) \times (0,\infty) \times [0,1] \times (0,\infty) \times [0, \infty)\times (0,\infty) \times [0,\infty) \times [0,\infty), 
\end{align}
given by
\[
m_n(A) = \frac{1}{n}\sum_{i=1}^n1_A(\zeta_i), \ \text{ for Borel sets } A \subset \cZ^e .
\]
Assume that as the number of agents becomes large, $n\rightarrow\infty$, the above empirical measure $m_n$ has a weak limit $m$ in the sense that $\int_{\cZ^e } f\,dm_n \rightarrow \int_{\cZ^e } f\,dm$ for every bounded continuous function $f$ on $\cZ^e $. For example, this holds almost surely if the $\zeta_i$'s are i.i.d.\ samples from $m$. Let $\zeta=(\xi,\delta,\theta,\epsilon,\rho,\breve\mu,\nu,\sigma)$ denote a $\cZ^e$-valued random variable with  limiting distribution $m$.

The \emph{mean field game} (MFG) we define next allows us to derive the limiting strategy as the outcome of a self-contained equilibrium problem, which intuitively represents a game with a continuum of agents with type distribution $m$. Rather than directly modelling a continuum of agents, we follow the MFG paradigm of modelling a single generic agent who we view as randomly selected from the population. 
The probability measure $m$ represents the distribution of type parameters among the continuum of agents. Equivalently, the generic agent's type vector is a random variable with law $m$. Heuristically, each agent in the continuum trades in a single stock driven by two Brownian motions, one of which is unique to this agent and one of which is common to all agents. We extend the Forward Nash equilibrium of Definition \ref{def:ForwardNashEquilibrium-consumption} to the MFG setting below.

\subsection{Agents as type-distribution samples  and the market}
\label{subsec:agents-as-type}
Let $(\Omega, \cF,\bF = (\cF)_{t \geqs 0},\bP)$ be a stochastic basis supporting two independent Brownian motions $W = (W_t)_{t \geqs 0}$ and $B = (B_t)_{t \geqs 0}$ together with a random vector $\zeta$ having distribution $m$ and given by
\begin{align}
    \label{eq:type-vector-mfg}
    \zeta = (\xi,\delta,\theta,\epsilon,\rho,\breve\mu,\nu,\sigma),
\end{align}
with values in the space $\cZ^e $ defined in \eqref{def:exp-type-space}, and
independent of $W$ and $B$.
Let $\bF = (\cF_t)_{t \in [0,T]}$ denote the smallest filtration satisfying the usual assumptions for which $\zeta $ is $\cF^{\mathrm{MF}}_0$-measurable and both $W$ and $B$ are adapted. Let also $\bF^B=(\cF^B_t)_{t \in [0,T]}$ denote the natural filtration generated by the Brownian motion $B$.

The \emph{generic agent's} wealth process is
\begin{align}
\label{def:X-MFG-consumption}
dX_t = rX_t + \pi_tX_t\big((\breve\mu-r) dt + \nu dW_t + \sigma dB_t\big) - c_t X_t dt, \quad X_0 = \xi, 
\end{align}
for a self-financing strategy, $(\pi_{t})_{t \geqs 0}$, standing for the fraction of wealth invested in the risky asset and consumption policy $(c^i_t)_{t \geqs 0}$, representing the instantaneous rate of consumption per unit of wealth. Together they must belong to the admissible set 
\begin{align*}
    \cA_{\mathrm{MF}} = \Big\{(\pi,c)&: 
    \textrm{$\bF$-progressively measurable $\bR \times (0,\infty)$-valued process } (\pi_t,c_t)_{t \geqs 0},
    \\
    & \qquad
        \bE\Big[\int_0^t (|\pi_s|^2 + |c_s|^2 ) ds \Big]< \infty,\
     \textrm{for any } t>0
    \Big\}.
\end{align*}
The risk-free interest rate $r$ is deterministic and fixed for the entire population. The random variable $\xi$ is the initial wealth of the generic agent, whereas $(\breve\mu,\nu,\sigma)$ are the market parameters. As in the previous sections, we denote $\mu := \breve\mu - r$ as an excess return. In the sequel, the parameters $\delta,~\theta, ~\epsilon$ and $\rho$ will affect the risk and consumption preferences of the generic agent. Each agent among the continuum will have different preference parameters. Hence these eight parameters are $\cF_0$-random, and each has the same interpretation as in the $n$-player game of the earlier section.

\subsection{The mean field equilibrium}

The formulation of the forward Nash game of Section \ref{sec:FRPCwithaverage} drives the formulation of the mean field game we discuss here, see also \cite{lackersoret2020many,platonov2020forward}. Recall that in the MFG-formulation, the \emph{generic agent} does not influence the average wealth of the continuum of agents, as but one agent amid a continuum.

We next introduce the concept of \emph{mean field (MF)-forward relative performance}, $\pi^*,c^*$ is  the \emph{MF-equilibrium} and, the main object of interest the \emph{MF-Forward relative performance equilibrium}.
\begin{assumption}
\label{ass:regularity-u-consumption-mfg}
Assume the derivatives $U_t(x,t)$, $U_{x}(x,t)$, $U_{xx}(x,t)$ and $V_{x}(y,t)$, $V_{xx}(x,t)$ exist for all $t \geqs 0, ~x > 0,~\bP$-a.s. and furthermore, $x \mapsto U(x,t)$ and $x \mapsto V(x,t)$ are strictly increasing ($U_x,V_x > 0$) and strictly concave ($U_{xx},V_{xx} < 0$) for any $t \geqs 0,~ x > 0,~ \bP$-a.s.
\end{assumption}
Given this setup we next define our concept of equilibrium (see also  \cite{platonov2020forward}).
\begin{definition}[MF-CRRA-Forward relative performance equilibrium (for the generic manager)]
\label{def:MFG-Forward-problem-consumption}
Let $(\overline{X}_t)_{t\geqs 0}$ and $(\overline{\Gamma}_t)_{t\geqs 0}$ be $\bF^{B}$-adapted positive square integrable stochastic processes representing the geometric average wealth and geometric average consumption of the continuum of agents respectively. Let $(\pi,c)\in \cA^{\mathrm{MF}}$ and $X^{\pi,c}$ solve \eqref{def:X-MFG-consumption} with $\pi$ and $c$. 

Set a $\bF^{\mathrm{MF}}$-progressively measurable random field $(Q(x,t))_{t\geqs 0}$ having, for some $\rho \geqs 0$, the dynamics 
\begin{align}
\label{eq:StructureQgeneral}
    Q(x,t) & = e^{-\rho t} U(x,t) + \int_0^t e^{-\rho s} V(\hat c_s x,s)ds,
\end{align} 
where $\hat c_t = c_t \overline{\Gamma_t}^{-\theta}$ and 
$U,V:\Omega\times (0,\infty)\times [0,\infty)\to \bR$ are two other $\bF$-progressively measurable random fields.

The field $Q$ is an \emph{MF-forward relative performance} for the \emph{generic manager} if, for all $t\geqs 0$, the following conditions hold:
\begin{itemize}
	\item[i)] The mappings $x \mapsto U(x,t),~ x \mapsto V(x,t)$, are $\bP$-a.s.~strictly increasing and strictly concave;
	\item[ii)]  For any $(\pi,c) \in \cA^{\textrm{MF}}$, $Q( X_t^{\pi,c}\overline X_t^{-\theta},t)$ is a (local) supermartingale and $X^{\pi,c}$ is the \emph{generic agent's} wealth process solving \eqref{def:X-MFG-consumption} for the strategy $(\pi,c)$;
	\item[iii)] There exists $(\pi^{*},c^*)\in \cA^{\textrm{MF}}$ such that $Q( X^*_t \overline X_t^{-\theta},t)$ is a (local) martingale where $X^*$ solves \eqref{def:X-MFG-consumption} with $(\pi^{*},c^{*})$ plugged in as the strategy;
	\item[iv)] We call $(\pi^*,c^*)$ of iii) a \emph{MF-equilibrium} if 
	$$\overline X_t = \exp\bE[\log X^*_t| \cF^B_t]
	\quad\textrm{and}\quad 
	\overline \Gamma_t = \exp\bE[\log c^*_t| \cF^B_t]\quad \textrm{ for all $t\geqs 0$, $\bP$-a.s.,}
	$$
	where $X^*$ solves \eqref{def:X-MFG-consumption} with $(\pi^{*},c^*)$ plugged in as the strategy.
\end{itemize}
We denote the sextuple $(U,V,\pi^*,c^*,\overline X,\overline \Gamma)$ satisfying i)-iv) the \emph{MF-Forward relative performance equilibrium}.
\end{definition}

The last point can be understood as a fixed point argument which creates a compatibility condition between the generic agent within the continuum of agents. Conditionally on the Brownian motion $B$, each agent faces an independent noise $W$ and an independent type vector $\zeta$. As argued in  \cite{LackerZariphopoulou2017,lackersoret2020many,platonov2020forward}, conditionally on $B$, all agents faces i.i.d. copies of the same optimisation problem. The law of large numbers suggests that the geometric average terminal wealth of the whole population should be $\exp\bE[\log X_t^*|\cF^B_t]$.

Our construction allows us to identify $\exp\bE[\log X_t^*|\cF^B_t]$ with a certain dynamics and, in turn, treat this component as an additional uncontrolled state process. This avoids arguments using the master equation for models with different types of agents altogether. 

It is unclear if the definition for forward MFG we propose is sufficiently tractable (concretely point iv)) to address the case of random coefficients or the general setting of a forward utility map with a full It\^o-field representation (with volatility). This same issue is already present in the earlier works  \cite{LackerZariphopoulou2017,lackersoret2020many} if one wanted to generalise to random coefficients there. Nonetheless, the definition works very well within the Merton market model setup, yielding a particularly tractable framework of study to understand the inherent difficulties of the general case. Additionally, our goal is to interpret the results further as this is one of the first works that combine MFGs with forward performance criteria.

\subsection{Solving the optimisation problem}

We now present the main result of this section. The existence of a \emph{MF-Forward relative performance equilibrium} for the generic manager according to Definition \ref{def:MFG-Forward-problem-consumption} within the context of a structural assumption similar to \eqref{eq:Solution Postulation-consumption}.

From the methodological point of view, the problem is solved as before. Apply It\^o's formula to $Q(Z^{\pi,c}_t,t)$, determine the optimal strategy $(\pi^*,c^*)$ and the consistency condition (the PDE) for $Q$ such that the first three conditions of Definition \ref{def:MFG-Forward-problem-consumption} hold. The last condition, to show that $(\pi^*,c^*)$ is indeed the MFG Forward equilibrium will follow by construction as the whole formulation is built with the fixed-point condition embedded from the start. 

We next present the \textit{separable power factor form}  assumption for $U$ and $V$.
\begin{assumption}
\label{ass:form-of-U-V-mean-field-g-with-f}
Let $U(x,t),V(x,t)$ admit the following form, for $x > 0,~t\geqs 0$,
\begin{align*}
    U(x,t) = \frac1{1-\frac1{\delta}}x^{1-\frac1{\delta}}f(t), \qquad V(x,t) =\frac1{\epsilon} \frac1{1-\frac1{\delta}}x^{1-\frac1{\delta}}g(t),
\end{align*}
for $t \mapsto f(t)$ differentiable for any $t \geqs 0$ and $f(0) = 1$, for some random variable $\epsilon > 0$ and $g(t) = f(t)^{1-\kappa}$ for some $\kappa \in \bR$ (where we have $g(0)=1^{1-\kappa} = 1$). Moreover, let for any $t \geqs 0,~\int_0^t |f(s)|^2 ds < \infty,~\bP$-a.s.
\end{assumption}

The parameter $\delta>0,\delta\neq 1$ represents the agent's risk tolerance while $\epsilon>0$, associated to the consumption element, represents the relative importance the agent assigns to the consumption when compared to the wealth. The parameter $\kappa$ represent the \emph{market-risk relative consumption preference} (in Section \ref{sec:examplesInterpretations} we discuss its economic features). 

As in the previous section, this assumption is a natural one to make in order to solve \eqref{eq:MF-SPDE-consumption} below (see additionally \cite{elkaroui2018consistent}, \cite[Definition 2.9]{avanesyan2020construction} or \cite[Section 4]{chong2018optimal}).

\begin{theorem}
\label{theo:MFG-solution-consumption}
Assume that $\delta>0,~\theta\in[0,1],~\mu>0,~\sigma\geqs 0,~\nu\geqs 0$ such that $\sigma^2+\nu^2>0$. 
Define constants $\psi^\sigma,\varphi^\sigma,\psi^\mu,\varphi^\mu$, assume they are finite and have explicit form given by
\begin{align*}
\psi^\sigma&= \bE\Big[\theta(1-\delta)\frac{ \sigma^2}{\nu^2+\sigma^2}\Big],
\varphi^\sigma= \bE\Big[\delta \frac{\mu\sigma }{\nu^2+\sigma^2}\Big],
\psi^\mu
=\bE\Big[\theta(1-\delta)\frac{ \mu \sigma}{\nu^2+\sigma^2}\Big],
\varphi^\mu=\bE\Big[ \delta \frac{\mu^2 }{\nu^2+\sigma^2} \Big].
\end{align*}
Assume further $\psi^\sigma\neq 1$ and define for some $\kappa \in \bR$, 
\begin{align}
\label{eq:mfg-equilibrium-pi-c}
    \pi^* 
    &=
    \frac1{{\nu^2+\sigma^2}}\Big( \theta (1-\delta)\sigma \frac{\varphi^\sigma}{1-\psi^\sigma}  +  \mu \delta \Big),\quad
    c^*_t =  
    \begin{cases}
        \bigg(\frac{1}{\beta} + \Big(\frac1{\lambda} - \frac1{\beta}\Big)e^{-\kappa\beta t} \bigg)^{-1}, & \beta \neq 0,\\
        \big(\kappa t + \frac1{\lambda}\big)^{-1}, & \beta = 0,
    \end{cases}
\end{align}
where 
\begin{align*}
    \lambda&=\epsilon^{-\delta}\big(e^{\bE[\delta\log \epsilon]}\big)^{\frac{\theta(1-\delta)}{\bE[ \theta(1-\delta)]-1}}, 
    \qquad 
    \beta = \frac{\theta(1-\delta)}{\bE[ \theta(1-\delta)]-1}\bE[\eta \delta] -\eta\delta,
\end{align*}
and
\begin{align}
    \label{eq:rho-mfg}
    \eta =  \Big(1  
    -\frac1{\delta}\Big) & \bigg(\frac{\delta\big(\mu-\theta\sigma\overline{\sigma\pi^*}(1  -\frac{1}{\delta})\big)^2}{2(\nu^2+\sigma^2)} 
    -r(1-\theta)
    \\ \nonumber 
    &\qquad \qquad -\theta_i\overline{\mu\pi^*} + \frac{\theta}2 \overline{\Sigma (\pi^*)^2} + \frac{\theta^2}{2} \big(\overline{\sigma\pi^*}\big)^2 \big(1-\frac1{\delta}\big) \bigg) - \rho,
\end{align}
where $\overline{ \sigma \pi^*},~\overline{ \mu \pi^*}$ and $\overline{\Sigma(\pi^*)^2}$ are given by \eqref{eq:sigma-pi-pverline}, \eqref{eq:ConsistencyExpressions-alphasigma-alphamu} and \eqref{eq:ConsistencyExpressions-Sigmapi} respectively.

Take the field $Q$ of \eqref{eq:StructureQgeneral} with $U,V$ satisfying Assumption \ref{ass:regularity-u-consumption-mfg} and \ref{ass:form-of-U-V-mean-field-g-with-f} (for the $\kappa$ above), and the consistency PDE 
\begin{align}
\label{eq:MF-SPDE-consumption}
\nonumber
U_t 
&= 
\theta \bigg( 
\frac{\varphi^\sigma}{1- \psi^\sigma}\psi^\mu+\varphi^\mu
- \frac{\sigma \frac{\varphi^\sigma}{1-\psi^\sigma}(\mu - \theta\sigma \frac{\varphi^\sigma}{1-\psi^\sigma})}{\nu^2+\sigma^2} -r(1-\theta) - \frac\theta2 \overline{\Sigma(\pi^*)^2}-\frac{\theta}2\Big(\frac{\varphi^\sigma}{1-\psi^\sigma}\Big)^2
  \bigg) xU_x
\\
&
\qquad \qquad 
+  \rho U + \frac{\big(\mu -\theta\sigma \frac{\varphi^\sigma}{1- \psi^\sigma}\big)^2}{2(\nu^2+\sigma^2)}   \frac{(U_x)^2}{U_{xx}}
+ \frac12\theta^2 
\Big(\frac{\varphi^\sigma}{1-\psi^\sigma}\Big)^2
\Big(\frac{\sigma^2}{\nu^2+\sigma^2}-1 \Big)x^2U_{xx}
\\ \nonumber
& \qquad \qquad  + \theta U_x \bar c_t- \widetilde {V} (U_x,t),
\end{align}
where $\overline{\Sigma(\pi^*)^2}$ is given by \eqref{eq:ConsistencyExpressions-Sigmapi}, $\bar c_t = \bE[c^*_t]$, where $c^*$ is from \eqref{eq:mfg-equilibrium-pi-c} and $\widetilde{V}$ is Fenchel-Legendre transform of $V$ (with initial condition of power type according to Assumption \ref{ass:form-of-U-V-mean-field-g-with-f} with $t=0$).

Then, if $\kappa < 0$ or $\kappa > 0,~ \beta > \lambda$ there exists a unique (parameterised by $\kappa$) MF-Forward CRRA relative performance equilibrium in the sense of Definition \ref{def:MFG-Forward-problem-consumption}. The MF-equilibrium strategy is given by $(\pi^*,c^*)$ from \eqref{eq:mfg-equilibrium-pi-c} and MF-forward CRRA relative performance utility is given by $Q$, for $U,V$ satisfying \eqref{eq:MF-SPDE-consumption} and $\overline X_t = \exp{\bE \big[\log(X_t^{\pi^*,c^*})\big]},~\overline \Gamma_t = \exp{\bE [\log(c_t^*)]}$.
Finally, utility time dynamics $f$ and $g$ satisfy
\begin{align}
    \label{eq:f-g-final-form-mfg}
    f(t) = \begin{cases}
        \Big(\big(c_t\big)^{\frac1{\delta}}\big(\tilde{c}_t\big)^{\theta(\frac1{\delta}-1)}\epsilon\Big)^{-\frac1{\kappa}}, & \kappa \neq 0,\\
        \exp\Big\{\left(\theta(\frac1\delta - 1) \bE[\lambda] - \eta -\frac{\lambda}{\delta}\right) t\Big\}, &\kappa = 0,
    \end{cases} \quad \text{and} \quad g(t) = f(t)^{1-\kappa}. 
\end{align}
If $\psi^\sigma=1$ {or $\kappa > 0, \beta \leqslant \lambda$}, then there exists no MF-equilibrium.
\end{theorem}
\begin{proof}
The complete proof can be found in Section \ref{supsec:proofMFgame} of the  appendix.
\end{proof}



We address the admissibility of $c^*$ in Section \ref{sec:examplesInterpretations}.

We now provide the single stock case result.
\begin{corollary}[Single stock]
\label{cor:mf-single-stock}
Let $\mu,\sigma,\nu$ be deterministic with $\nu  = 0$ and $\mu,\sigma > 0$. Then $\lambda,\beta$ are expressed as 
\begin{align}
    \label{eq:beta-lambda-mfg-single-stock}
    \lambda&=\epsilon^{-\delta}\big(e^{\bE[\delta\log \epsilon]}\big)^{\frac{\theta(1-\delta)}{\bE[\theta(1-\delta)]-1}}, \\ \nonumber
    \beta &= \bigg(\frac{\mu^2}{2\sigma^2}\Big(\delta + \frac{\theta}{\theta_{\textrm{crit}}}(1-\delta)\Big) + r \bigg)(1-\delta)\big(1 -\frac{\theta}{\theta_{\textrm{crit}}}\big) +  \frac{\theta}{\theta_{\text{crit}}}\frac{\bE[\rho \delta]}{\bE[\delta]}(1-\delta) +\rho\delta ,
\end{align}
where $\theta_{\textrm{crit}} = \frac{1-\bE[\theta(1-\delta)]}{\bE[\delta]}$.

If $\bE[\theta(1-\delta)]\neq 1$ then a {optimal candidate strategy exists}, with the optimal strategy $(\pi^{*},c^{*})$ given by
\begin{align}
    \pi^{*} = \frac{\mu}{\sigma^2} \Big(\delta + \frac{\theta}{\theta_{\textrm{crit}}}(1-\delta) \Big)
    \ \textrm{ and }\
    \label{eq:mfg-single-stock-optimal-consumption}
    c^*_t =  
    \begin{cases}
        \bigg(\frac{1}{\beta} + \Big(\frac1{\lambda} - \frac1{\beta}\Big)e^{\kappa\beta t} \bigg)^{-1}, & \beta \neq 0,\\
        \big(-\kappa t + \frac1{\lambda}\big)^{-1}, & \beta = 0.
    \end{cases}
\end{align}
Furthermore, {the optimal candidate strategy} is a MF-equilibrium, if either $\kappa \leqs 0$, or $\kappa > 0,~ \lambda <  \beta$.
\end{corollary}

\begin{remark}[On convergence of the Nash equilibria of $n$-player game]
    We can see that the MF-equilibrium agrees with the limit of the $n$-player game equilibrium strategies, as $n \to +\infty,~\bP$-a.s.. The respective parameters and their functions converge to the averages of type distributions by the strong law of large numbers.
\end{remark}
%
%
\section{A discussion of the results}
\label{sec:examplesInterpretations}

In this section, we focus on the interpretation of the results obtained by way of analyzing the single stock case of the mean field game given as by Corollary \ref{cor:mf-single-stock} (compare with Corollary \ref{cor:single-stock-n-player}). This approach provides already valuable insights on the model, capturing the relations and their interpretation whilst avoiding the complex, intricate dependencies of the general case. This case also allows for easier comparison with the results from \cite{lackersoret2020many} who used standard utility maps.
For the sake of simplicity, if not specified, we write $\pi,c$ to refer to $(\pi^*,c^*)$ the mean field single stock {optimal candidate strategy} given by \eqref{eq:mfg-single-stock-optimal-consumption}. 
{For a full overview, we keep mention of the inadmissible equilibria cases, even though our discussion focuses on the MF-equilibrium strategy.}
We do not discuss the optimal investment strategy $\pi$ and omit some details on the optimal consumption strategy $c$ since these overlap with preceding work  \cite{lackersoret2020many,LackerZariphopoulou2017,platonov2020forward}. 
We restrict our attention to the \emph{market-risk relative consumption preference} $\kappa$ and its interplay with the other parameters in the scope of the optimal consumption $c$. The case $\kappa=1$ is just that of \cite[Section 4]{lackersoret2020many} and we point the reader to the discussion there. Before starting, we emphasize that the behaviours found here, when $\kappa\neq 1$, are not a symmetrised version of those of \cite{lackersoret2020many} -- this is easily seen in Figure \ref{fig:plots-of-c-for-various-parameters} and Table \ref{table:consumptionSignAsMapofParameters}.

\subsection*{Classical parameters}

\textit{Recovering classical results.} Under no performance concerns and no competition (e.g., $\theta = 0$) we recover the classical results by Merton \cite{merton1969lifetime} and Samuelson \cite{samuelson1969lifetime} in \eqref{eq:mfg-single-stock-optimal-consumption}: agents invest a constant fraction of wealth in the stock, the consumption strategy is time-dependent, and  both investment and consumption strategies are independent -- this holds for the classic utility results in \cite{LackerZariphopoulou2017}. Both here and in \cite{platonov2020forward} the results differ from the mentioned ones as the classical results encapsulate a dependence on the time horizon. When $\kappa=1$ we recover the results in \cite{lackersoret2020many}.

\smallskip
\textit{The market parameters $\beta, \lambda$ in \eqref{eq:beta-lambda-mfg-single-stock}.} As in \cite{lackersoret2020many} and Corollary \ref{cor:mf-single-stock} we rewrite $\beta$ in \eqref{eq:beta-lambda-mfg-single-stock} as
    \begin{align*}
        \beta 
        = \frac{\mu^2}{2\sigma^2}(1-\delta_{\text{eff}})\delta_{\text{eff}} &+ {r(1-\delta_{\text{eff}})} + {\rho \delta^\prime_{\text{eff}}},
        \\
        \quad\textrm{where}\quad 
        \delta_{\text{eff}} 
        = (1-\frac{\theta}{\theta_{\text{crit}}})\delta + \frac{\theta}{\theta_{\text{crit}}}
        \quad &\textrm{and}\quad 
        \delta^\prime_{\text{eff}}= (1-\frac{\theta}{\theta^\prime_{\text{crit}}})\delta + \frac{\theta}{\theta^\prime_{\text{crit}}}
        \\
       \quad\textrm{for}\quad 
       \theta_{\textrm{crit}} 
       = \frac{1-\bE[\theta(1-\delta)]}{\bE[\delta]} 
       \quad &\textrm{and}\quad 
       \theta^\prime_{\textrm{crit}} = \rho\frac{1-\bE[\theta(1-\delta)]}{\bE[\delta \rho]}.
    \end{align*}
    When $\rho$ is deterministic across the population we have $\delta_\text{eff} = \delta^\prime_{\text{eff}}$ and $\theta_{\textrm{crit}}=\theta^\prime_{\textrm{crit}}$. Furthermore, $\beta$ reduces to the following expression
    $\beta =\frac{\mu^2}{2\sigma^2}(1-\delta)\delta + r(1-\delta) + \rho \delta$, which is reminiscent of the classical Merton result.
    Additionally, setting $r=\rho=0$, we can rewrite $\beta$ as
    \begin{align*}
        \beta=\frac12\pi\mu (\delta-1)\Big(\frac{\theta}{\theta_{\textrm{crit}}}-1\Big),
    \end{align*}
    which shows that $\beta$ is, in fact, the expected return times some additional risk (on top of the log risk-tolerant investor) and a different competition proportion (on top of $\theta_{\text{crit}}$-competitive investor).  
    We hence interpret $\beta$ as the \emph{expected effective portfolio return}. 
    
    When it comes to $\lambda$ (see \eqref{eq:beta-lambda-mfg-single-stock}),
    \begin{align*}
            \lambda
            &
            =\frac{1}{\epsilon^{\delta}} 
            \times (H_{\textrm{population}})^{\theta(1-\delta)}
            \quad \textrm{where} \quad 
            H_{\textrm{population}}= \big(e^{\bE[\delta\log \epsilon]}\big)^{\frac{1}{\bE[\theta(1-\delta)]-1}},
    \end{align*}
    the main relation to notice is its inverse dependence on the \textit{agent's relative perception of wealth relative to consumption} $\epsilon$ and scaled with the risk tolerance $\delta$, and then adjusted 
    by the agent's specific weighting (associated to the agent's risk aversion and competition preferences) of the population parameter  $H_{\textrm{population}}$.
    Hence, we treat $\lambda$ as an ``effective'' $\epsilon$. As originally $\epsilon$ captures the linear relation between utility from wealth and utility from consumption in the general FPP process (see Assumption 
    \ref{ass:form-of-U-V-mean-field-g-with-f}), we characterise $\lambda$ as the \emph{effective relative perception of wealth with respect to consumption}.

\subsection*{The market-risk relative consumption preference $\kappa$}

Before fully diving into the analysis of $\kappa$, we differentiate between two settings and their interpretations. An agent in the absence of competition ($\theta=0$) or purely looking for a ``best response'' to the actions undertaken by the other agents (see Remark \ref{rem:open-question-after-nash-n-player}) can determine its own idiosyncratic $\kappa$. At the same time, for the Nash forward equilibrium, all players must interact through a common $\kappa$. 
In either case, we interpret $\kappa$ as a \emph{market-risk relative consumption preference} observing that when the agent is free to choose $\kappa$ for herself ($\theta=0$ case or just best response), the agent looks for the most suitable market investment-consumption environment via the choice of $\kappa$.

\subsubsection*{The concept of Elasticity of intertemporal substitution (EIS)}

We start by calculating one of the main parameters in models of dynamic choice of consumption in macroeconomics and finance \cite{thimme2017intertemporal}. The \emph{elasticity of intertemporal substitution} (EIS) measures the agent's willingness to substitute future consumption for present consumption in response to changes in investment opportunities and is given by \cite{hall1988-eis,Azar2018-nexus-eis}
\begin{align*}
    \text{EIS} :
    &
    = -\dfrac{ d\big(\partial_t{c}_t/c_t\big)}{d\big(\partial_t V_x(c_tz,t)/V_x(c_tz,t)\big)},
\end{align*}
for some\footnote{The formula presented has stated the agent's rate of consumption equal to $c_tz$, where $c_t$ is a rate of consumption for unit of wealth and $z > 0$ some fixed amount of wealth. On the one hand, fixing the amount of wealth helps us to purely investigate the relative effect of consumption, on the other, it allows to assess its time differential to calculate $\textrm{EIS}$, in contrast with $c_t X_t$ having an It\^o-differential form and thus not computable.} $z > 0$. 
EIS measures the response of consumption growth to the deviation of \emph{real interest rate}, the latter captured by the denominator of the above expression. 
If $\text{EIS} < 0$, the agent is subjected to \emph{income effect}. 
On the contrary, $\text{EIS} > 0$ is known as \emph{substitution effect} (see introduction). Finally, $\text{EIS} = 0$ corresponds to a case of constant consumption rate where the agent is indifferent about the real-time growth and would like to spend the constant fraction of income all over the time.

We introduce the reader to the new concept of \emph{Elasticity of Conformity}.
\begin{definition}[Elasticity of Conformity (EC)]
    Let $c_t$ and $\tilde c_t$ be the consumption of the generic agent and geometric average \eqref{eq:c-hat-i} respectively. 
    We define \emph{Elasticity of Conformity} $\gamma_t$ as 
    \begin{align}
    \label{eq:EC-elasticity-of-conformity}
        \gamma_t := \dfrac{d\left(\partial_t c_t/c_t\right)}{d\left(\partial_t \tilde c_t/\tilde c_t\right)}.
    \end{align}
\end{definition}
The \emph{Elasticity of Conformity (EC)} captures how the change in the agent's consumption is affected by the deviation in the consumption of the typical agent (as compared against the geometric average) across time. The EC's negative/positive sign reflects that the agent's decision aligns with the primary trend of investors, i.e.,  to increase/decrease consumption in response to the decaying geometric average consumption. Moreover, the choice of the agent to maintain a non-zero level of consumption is captured by their EIS asymptotically going to $0$ as $t \to \infty$ (see Figure \ref{fig:all-plots-c-f-EC-EIS} a) and c)).

We now calculate the EIS of the generic agent
\begin{align*}
        \textrm{EIS} := -\dfrac{d\big(\partial_t{c}_t/c_t\big)}{d\big(\partial_t V_x(c_tz,t)/V_x(c_tz,t)\big)} 
        &= 
        -\dfrac{d \big(\partial_t c_t/c_t\big)}{d \Bigg(\dfrac{V_{xx}(c_tz,t)\partial_t{c}_tz+V_x(c_tz,t)\frac{\partial_t{g}(t)}{g(t)}}{V_x(c_tz,t)}\Bigg)}
        \\
        &
        = \dfrac{d \big(\partial_t c_t/c_t\big)}{d \Big(\frac1{\delta}\big(\partial_t c_t/c_t\big)-\big(\partial_t{g}(t)/g(t)\big)\Big)},
    \end{align*}
given the CRRA property $V_{x}/V_{xx} =-\delta x$. 
Injecting \eqref{eq:f-g-final-form-mfg}, $g(t) = \Big(\big(c_t\big)^{\frac1{\delta}}\big(\tilde{c}_t\big)^{\theta(\frac1{\delta}-1)}\epsilon\Big)^{\frac{\kappa-1}{\kappa}}$, yields
    \begin{align*}
        \frac{\partial_t{g}(t)}{g(t)} &= \frac{\partial_t{c}_t}{c_t}\frac{\kappa - 1}{\kappa}\frac1{\delta} + \frac{\partial_t{\tilde{c}}_t}{\tilde{c}_t}\frac{\kappa -1}{\kappa}\theta\left(\frac1{\delta}-1\right) 
    = \frac{\partial_t{c}_t}{c_t}\frac{\kappa-1}{\kappa}\frac1{\delta} \left(1 +  \theta\big(1 - \delta \big)\frac1{G_t}\right),
    \end{align*}
    where $G_t:= \dfrac{\partial_tc_t/c_t}{\partial_t\tilde{c}_t/\tilde{c}_t}$.
    Returning to the EIS expression from before, we find 
    \begin{align*}
        \textrm{EIS} 
        &= \dfrac{d \big(\partial_t c_t/c_t\big)}{d \Big(\frac1{\delta}\big(\partial_t c_t/c_t\big)-\frac{\kappa - 1}{\kappa}\frac1{\delta} \Big(1 + \theta\big(1 - \delta\big)\frac1G_t\Big)\big(\partial_t c_t/c_t\big)\Big)}
        = \frac{\delta\kappa}{1 - \theta(1-\kappa)(\delta - 1)\frac1{\gamma_t}}.
    \end{align*}

In contrast with the classic utility optimisation approach where one has $\textrm{EIS}^{\textrm{classic}}=\delta$ the FPP approach allows an intrinsic time-dependence of $V$ and thus of the EIS. Namely
\begin{align}
\label{eq:EIS-for-our-setting}
    \textrm{EIS}^{(\theta\neq 0)}_t 
    = \frac{\textrm{EIS}^{(\theta=0)}}{1 - \theta(1-\kappa)(\delta - 1)\frac1{\gamma_t} },
    \quad \textrm{where}\quad
    \textrm{EIS}^{(\theta=0)} = \kappa\delta = \kappa \times \textrm{EIS}^{\textrm{classic}},
\end{align}
with $\gamma_t $ standing for the time-dependent \emph{Elasticity of Conformity (EC)}. 

The FPP here allows one to disentangle risk tolerance $\delta$ and Elasticity of intertemporal substitution (EIS) -- standard utility theory does not allow this. This same feature has been reported in \cite[Slide 22]{bismuthgueant2019} through the use of recursive utilities (in a more restricted Merton market framework than ours but within  MFGs) and which is close to our result \eqref{eq:EIS-for-our-setting}. In contrast with standard utility theory, the FPP setting captures the different dimensions of risk, represented by $\kappa$, that purely comes from the consumption environment and further scales the EIS.
When $\kappa = 1$, the case of standard CRRA utility framework and analysed in \cite{lackersoret2020many}, immediately yields $\text{EIS} = \delta=\textrm{EIS}^{\textrm{classic}}$.

\begin{figure}[htb!]
        \centering 
        \includegraphics[width=0.9\textwidth]{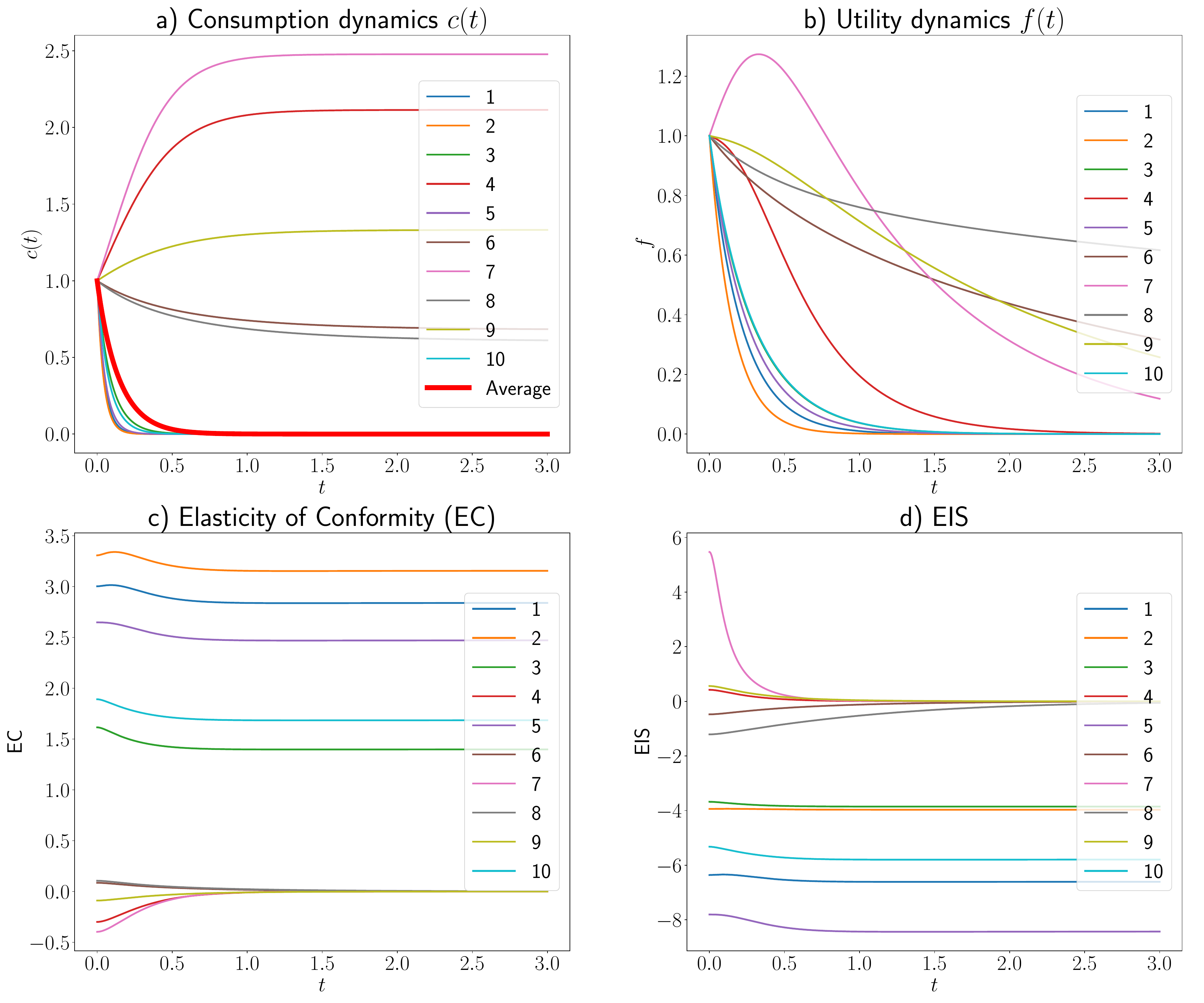}
        \caption{Simulated agents' types as a function of time: 
        a) Optimal consumption \eqref{eq:mfg-single-stock-optimal-consumption}; 
        b) Utility dynamics $f$ \eqref{eq:f-g-final-form-mfg}; 
        c) Elasticity of Conformity (EC) \eqref{eq:EC-elasticity-of-conformity}; 
        d) EIS \eqref{eq:EIS-for-our-setting}. 
        All the parameters are independent as a components of joint type's distribution and simulated as follows: $\delta/2 \sim Beta(1.5,1.5),~ \theta \sim Beta(3,5),~ \rho \sim Beta(1,19),~r = 0.05,~\lambda = 1,~ \mu = 5,~ \sigma = 1, ~\kappa = -2, ~ \theta_{\text{crit}} = 0.375$, where $Beta(\cdot,\cdot)$ refers to the Beta distribution.}
        \label{fig:all-plots-c-f-EC-EIS}
    \end{figure}

Figure \ref{fig:all-plots-c-f-EC-EIS} displays simulations of consumption, utility, EC and EIS for a collection of agents with a certain fixed distribution of the type vector. Looking into Figure \ref{fig:all-plots-c-f-EC-EIS}a) it is worth pointing out that, within our simulations, the average (geometric) consumption asymptotically decays to $0$ due to the behaviour of the majority of the agents being under the income effect. As $\beta$ is bounded from above but not from below, and given the symmetric distribution of type component for $\delta$ and $\theta$ around critical values of $1$ and $\theta_{crit}$ respectively, the simulated $\beta$ is much more likely to be negative and hence driving the consumption to $0$ on the long run.  
Agents that do not reduce their consumption to a zero constant level over time are indifferent to the main herd of investors, and one sees their Elasticity of Conformity (Figure \ref{fig:all-plots-c-f-EC-EIS}c)) converging to zero. 
Moreover, their consumption behaviour is becoming perfectly inelastic in the long run (EIS $\to 0$ as $t \to \infty$ in Figure \ref{fig:all-plots-c-f-EC-EIS}d)).

In contrast, agents that choose to gradually cut their consumption to $0$ (following the market trend) have a constant positive EC and a constant negative EIS in the long term.

\subsubsection*{The market-risk relative consumption preference $\kappa$ under no performance concerns: $\theta = 0$ (no competition)}

\emph{Choosing different EIS and Risk Tolerance.}
In the absence of competition, i.e., when $\theta = 0$,  \eqref{eq:EIS-for-our-setting} already yields a non-standard $\textrm{EIS}^{(\theta = 0)} = \kappa \delta$, whereas standard utility theory with CRRA type utilities yields $\textrm{EIS}^{\textrm{classic}} = \delta$. Without competition, as we no longer require $\kappa$ to be uniform and deterministic across the whole population, the agent can choose in which consumption environment he is willing to invest, hence, as opposed to \cite{lackersoret2020many}, the agent can determine her own risk tolerance $\delta$ and $\textrm{EIS}^{(\theta = 0)}$ independently.

For $\kappa > 0$ then $\textrm{EIS}^{(\theta = 0)} > 0$ and hence the agent is subjected to \emph{income} effect, otherwise, when $\kappa < 0$ then $\textrm{EIS}^{(\theta = 0)} < 0$ and the agent allows for a \emph{substitution} effect. Finally, $\kappa = 0$ reflects a constant consumption rate as $\textrm{EIS}^{(\theta = 0)} = 0$. This last scenario is particularly important as empirical research finds $\textrm{EIS}$ in general to be around $0$ for the average household (see \cite{hall1988-eis}, \cite{CampbellMankiw1991},  \cite{cundy2018essays}). In our framework, $\textrm{EIS}= 0$ may simply mean $\kappa= 0$ instead of a low risk tolerance $\delta = 0$.

{We do not present any formal proof but argue that $\kappa$ cannot be made to disappear from the problem at the expense of redefining the remaining model parameters (namely, $r, \rho, \nu, \epsilon,\beta,\lambda$) and then referring to a variant of the model  (with the redefined parameters). In effect, $\kappa$ represents a new degree of freedom in modelling.} Nonetheless, a scaling effect between $\kappa, \beta,\lambda$ at the level of consumption strategy is present.

Concretely, given $\beta$ and $\lambda$, define $\hat \beta := \kappa \beta$, $\hat \lambda := \kappa \lambda$, and $c^{(\beta,\lambda,\kappa)}_t:=c^*_t$ with $c^*_t$ the expression given in \eqref{eq:mfg-single-stock-optimal-consumption} for $\kappa \neq 0$ (and $\beta,\lambda$). Then $c^{(\beta,\lambda,\kappa)}$ can be rewritten as follows
\begin{align*}
    c^{(\beta,\lambda,\kappa)}_t = \kappa^{-1} c^{(\hat\beta,\hat\lambda,\kappa=1)}_t.
\end{align*}

It is important to remark that for $\kappa=1$ the consumption  $c^{(\hat\beta,\hat\lambda,\kappa=1)}$ is that of  \cite{lackersoret2020many}. 
By construction $\lambda>0$ but $\hat \lambda\in\bR$, on the other hand, for $\kappa \geqs 0$ the parameters $\hat \beta$, $\hat \lambda$ preserve the original sign of $\beta$ and  $\lambda$.

\subsubsection*{The market-risk relative consumption preference $\kappa$ under performance concerns: $\theta\neq 0$}

The agents are now synchronised on their consumption preference as they compete ($\kappa$ is the same for all agents, see Remark \ref{rem:open-question-after-nash-n-player}). From \eqref{eq:EIS-for-our-setting} the agent sees her consumption dynamics being re-scaled by competition ($\theta \neq 0$) and the market environment ($\kappa \neq 0)$. The agents can access their $\text{EIS}$,  \eqref{eq:EIS-for-our-setting}, by adjusting the risk-competition preferences (the EC depends on $\beta$ and $\lambda$ as implied by the agents' parameters).

\begin{figure}[htb!]
        \centering 
        \includegraphics[width=0.9\textwidth]{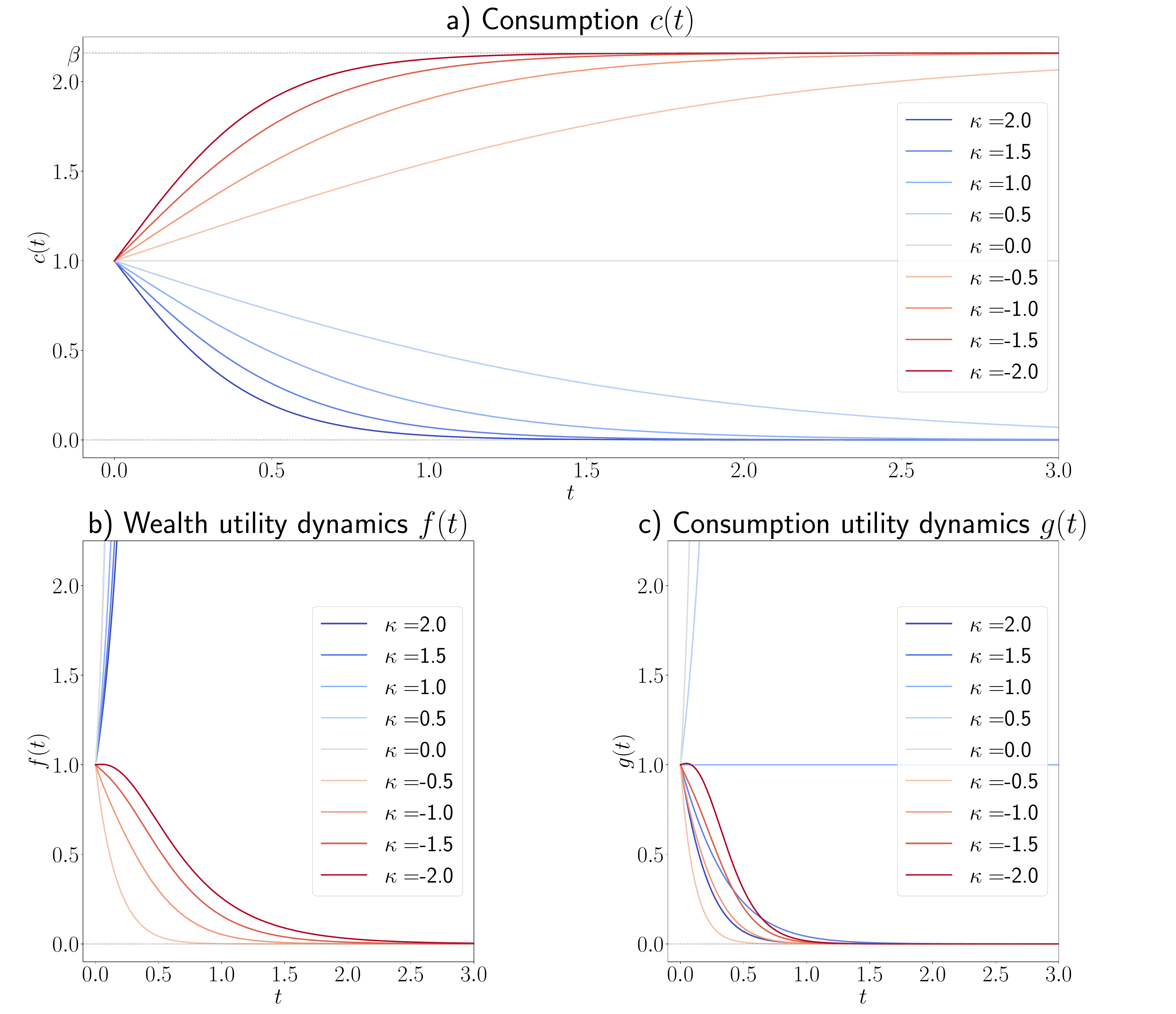}
        \caption{Plots of different elements from a single agent's optimization problem for $\kappa$ from $-2$ to $2$: a) Consumption dynamics $c(t)$ with lower asymptote at Zero and upper asymptote at $\beta$; b) Wealth utility dynamics $f(t)$; c) Consumption utility dynamics $g(t)$. The parameters are $\delta = 0.5,~ \theta = 0.6,~ \rho = 0.04,~r = 0.05,~\beta = 2.16,~\lambda = 1,~ \mu = 5,~ \sigma = 1$.
        }
        \label{fig:plots-of-f-g-for-different-kappas}
    \end{figure}

\textit{The ``market-risk relative consumption preference'' parameter $\kappa$.} 
This parameter appears in Assumption \ref{ass:form-of-g-n-player-example} and \ref{ass:form-of-U-V-mean-field-g-with-f} to weigh the utility of wealth versus consumption through $g(t)=f(t)^{1 - \kappa}$. It reflects the weight in time the environment assigns to the consumption process compared to the accumulation of wealth. In parallel with the risk tolerance $\delta$ of the standard utility $u_0(x) = \frac{1}{1-\frac1\delta} x^{1-\frac1\delta}$, we interpret $\kappa$ as a parameter of market risk, reflecting how the market is encouraging its agents to prefer consumption over wealth and vice versa.
Recalling for simplicity the no-competition EIS, $\textrm{EIS}^{(\theta = 0)}= \kappa\delta$, we can see that $\kappa$ appears as another dimension of risk, further adjusting each agent's consumption response to the deviations in the real interest rate.

Figure \ref{fig:plots-of-f-g-for-different-kappas} displays the effect of $\kappa$ on consumption and utility dynamics for the generic agent. 
When $\kappa > 0$, the agent's utility is more sensitive to deviations of wealth than to consumption (as $g(t) = f(t)^{1-\kappa}$) and two cases need to be distinguished: $\kappa \geqs 1$ and $0 < \kappa < 1$.
The market environment for large $\kappa$ makes the agents' utility from wealth and consumption inversely related. In this scenario, the percent increase in utility from wealth will correspond to a decrease in utility from wealth and vice versa, making the agent \emph{consumption-to-wealth averse}.
On the other hand, when $0 < \kappa < 1$ the environment pushes the agent to derive a moderate utility from consumption relative to the utility from wealth. From a percent increase or decrease in utility from consumption, the agent will obtain a much larger respective change in her utility from wealth.

If $\kappa < 0$ the general trend is reversed from that of $0 < \kappa < 1$. The changes in consumption utility 
have a larger weight relative to those of the wealth utility dynamics.
We refer to all sub-cases of $\kappa < 1$ as \emph{consumption-to-wealth tolerant}, indicating that the agents' utilities from the accumulation of wealth and consumption are aligned. Still, we emphasise that the quality of their mutual relation is subjected to the magnitude of $\kappa$. Finally, $\kappa = 0$ prescribes $f = g$, and the environment tells the agent to prefer wealth and consumption equally. Here, $c_t=\lambda$ for all $t\geqs 0$ (see \eqref{eq:mfg-single-stock-optimal-consumption}), in other words, the agent has no dynamic preference to deviate from  the effective relative perception of wealth with respect to consumption $\lambda$. 
Looking back at \eqref{eq:mfg-single-stock-optimal-consumption}, one can rewrite $c$ as
$$
c_t=
\bigg(\frac{1}{\beta} + \Big(\frac1{\lambda} - \frac1{\beta}\Big)e^{\kappa\beta t} \bigg)^{-1}
=
\bigg(\frac{1}{\beta}\Big(1 - e^{\kappa\beta t}\Big) + \frac1{\lambda}e^{\kappa\beta t} \bigg)^{-1}.
$$
Roughly, with $\kappa=0$, the market does not have the agent competing their $\lambda$ against their expected effective portfolio return $\beta$.

Figure \ref{fig:plots-of-c-for-various-parameters} displays the effect of $\kappa$ on the consumption map for different combinations of $\lambda,\beta$. When $\kappa\in (-1,1)$ the steady-state of consumption (the asymptote) is reached much later. 
    \begin{figure}[htb!]
        \centering 
        \includegraphics[width=0.9\textwidth]{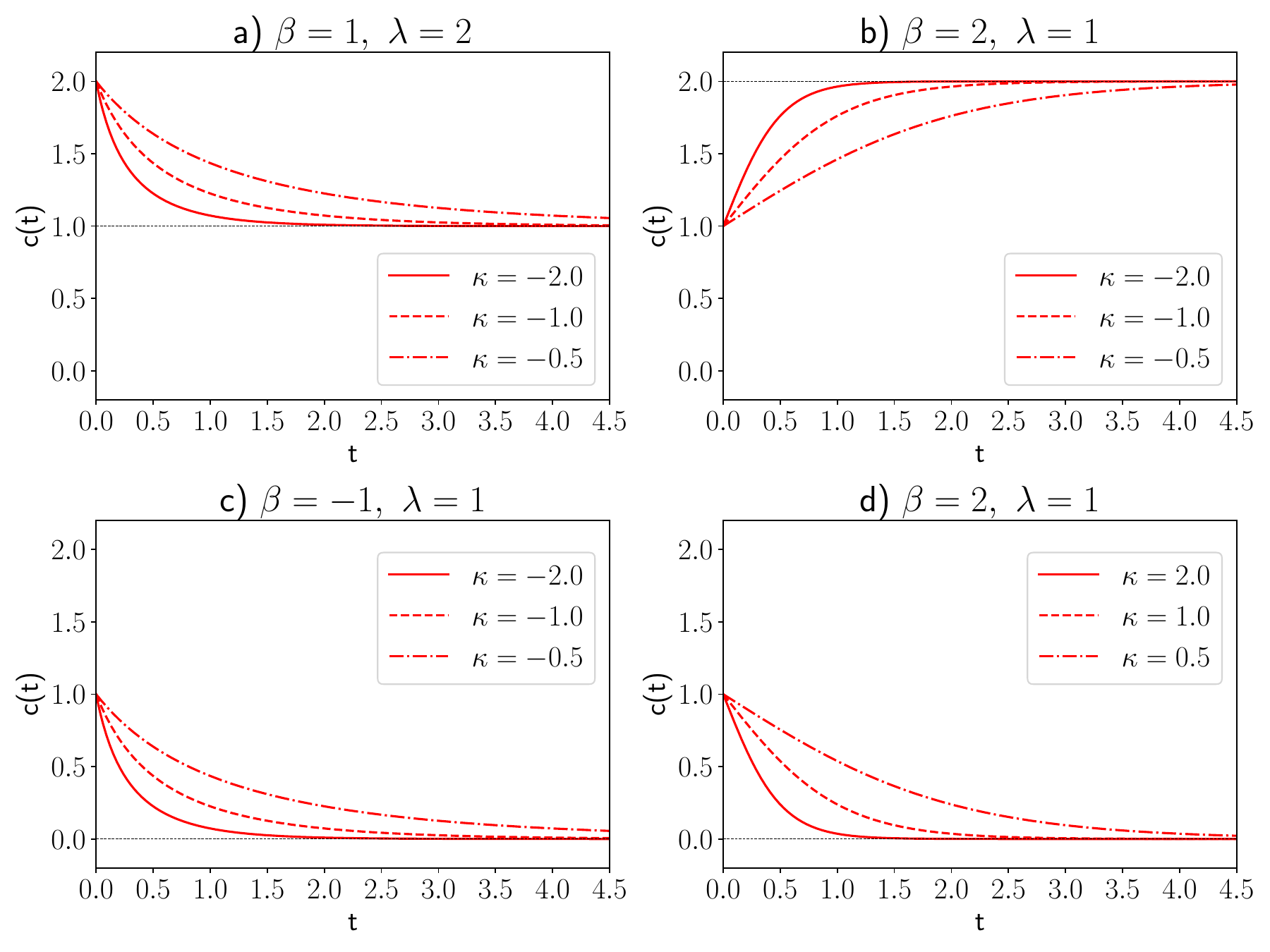}
        \caption{Plots of optimal $c$ of \eqref{eq:mfg-single-stock-optimal-consumption} for various combinations of parameters with a highlight on the monotonic behaviour of $\kappa \mapsto (c_t)_{t\geqs 0}$ by way in which the asymptote is reached.}
        \label{fig:plots-of-c-for-various-parameters}
    \end{figure}

In this model, a unique macro-economics cause-effect phenomenon of $\kappa$ is noteworthy. Assume a central planner that can control $\kappa$ then, given market parameters, by manipulating $\kappa$ the central planner can encourage or discourage consumption across the whole market. This is a well-known (Keynesian economics) policy for \textit{economic stimulus} in recession times. 

We briefly discuss the effect of the discounting factor $\rho$ and the riskless rate $r$. By looking into \eqref{eq:beta-lambda-mfg-single-stock}, one can notice the additive effect of $r$ and $\rho$ on the expected effective portfolio return. The sign of $\beta$ affects the direction and rate of the agent's consumption \eqref{eq:mfg-single-stock-optimal-consumption}. Thus,  by choosing the personal discount factor appropriately the agent can adjust for $\beta$, making it  larger/smaller, as to exceed the important thresholds of $\lambda$ and $0$ (that take place in $c^*$ \eqref{eq:mfg-single-stock-optimal-consumption} -- see Table \ref{table:consumptionSignAsMapofParameters}).

\textit{The finer interplay of $\kappa$ with $\beta$ and $\lambda$ on the consumption policy $c$.} 
We refer to \cite[Section 4]{lackersoret2020many} comprising the analysis for $\rho = r = 0$ and $\kappa = 1$. This reduction, upon inspection, extends to our case when $\kappa > 0$, thus we focus on the case where $\kappa < 0$, i.e., the market encourages the \emph{substitution} effect.
    
In terms of market behaviour, there are 4 cases to be distinguished: $\beta > \lambda$, $0 \leqs \beta < \lambda$, $\beta < 0 < \lambda$ and $\beta = \lambda$ and we summarise them in Table \ref{table:consumptionSignAsMapofParameters} (by monotonic behaviour, associated asymptotes, and admissibility). If $\beta = \lambda$ then $\kappa$ plays no role in $c$, but it does affect the utility map. We emphasise that the expected rate of return of wealth, ${\partial_t} \bE[\log X_t|\cF_0]$, follows the opposite direction of the monotonicity of $c$. 

First, $c_t$ is increasing in time, when $\beta > \lambda$. Indeed, when the relative perception of wealth is big, i.e., if $\epsilon \nearrow \infty$ then $\lambda \searrow 0^+$, the agent's consumption does not vanish and monotonically increases to $\beta$; this is in stark contrast with the case $\beta > \lambda$ under $\kappa > 0$, where $c_t \searrow 0^+$ monotonically as $t \to \infty$ (compare Figure \ref{fig:plots-of-c-for-various-parameters} b) and d).
From the beginning, the agent has a focus on wealth thus intends to consume less in the short-term (Figure \ref{fig:plots-of-c-for-various-parameters} b). However, in the long run, the agent follows the market preference towards consumption (see Figure \ref{fig:plots-of-f-g-for-different-kappas}, b) and c) for negative $\kappa$'s), and raises her proportion of consumption up to some asymptotically stable level $(=\beta)$. 
As the agent's wealth increases, she cannot resist the environment's push to spend more despite her preference for wealth. This particular situation mimics a ``keeping up with the Joneses'' \cite{Gali1994} behaviour given a population of ultra-rich agents: the effective utility from getting richer comes from increasing the level of consumption and is accelerated by competition and the intensive consumption environment.
    
The trend is similar to when $\beta < \lambda$, which is monotonically decreasing. Indeed, when $\epsilon \searrow 0^+,~\lambda \nearrow +\infty$, the agent prefers consumption over wealth (while the market pushes agents to consume $\kappa < 0$), so the agent consumes at a higher rate from the onset. However, we distinguish two sub-cases here. The agent with an expected effective portfolio return $\beta \geqs 0$, who seeks to stabilise her consumption and thus decrease it asymptotically up to the safe rate $\beta$ (Figure \ref{fig:plots-of-c-for-various-parameters} a).
When the agent chooses $\beta < 0$, she monotonically reduces her consumption rate to $0$, even if she had started from a higher level. That is the case when the agent is very risk-tolerant ($\delta > 1$) or very competitive ($\theta > \theta_{\textrm{crit}}$) or both, so the agent accumulates wealth even if she is deriving more utility-value from consumption over time (Figure \ref{fig:plots-of-c-for-various-parameters} c). 
    
Figure \ref{fig:plots-of-c-for-various-parameters} d) features the \emph{income} effect. The agent is pushed to asymptotically decrease her consumption to $0$, due to the environment promoting utility from wealth over utility from consumption -- see the positive $\kappa$ cases of Figure \ref{fig:plots-of-f-g-for-different-kappas} b) and c).
    
Finally, when $\kappa = 0$ or $\beta = \lambda$ the consumption $c_t$ is constant equal to $\lambda$. In the former scenario, $\kappa = 0$, the dynamics of the agent's preferences towards consumption and wealth is identical (captured by $f(t)=g(t)$ and thus $\text{EIS} = 0$). Hence, the agent's optimal consumption policy is the constant effective relative perception of wealth $\lambda$. 
Nonetheless, the agent's utility is affected by $\kappa$. All cases are neatly summarised in Table \ref{table:consumptionSignAsMapofParameters}. See also Figure \ref{fig:plots-of-monotonicity-c} in the supplementary material for the regions of monotonicity for $c_t$ as function of $\kappa,\delta,\theta$.

\begin{table}[htb!]
    \centering
    \caption{Behaviour of $(c_t)_{t\geqs 0}$ of \eqref{eq:mfg-single-stock-optimal-consumption} as a map of its parameters. Here $\searrow,\nearrow$ denote convergence from above/below respectively (as $t\to \infty$).
    Apparent is the asymmetry of cases for $\kappa < 0$ against $\kappa > 0$.}
    \label{table:consumptionSignAsMapofParameters}
    \begin{tabular*}{0.885\textwidth}{c|c|c|c|c|c}
    \toprule
    \multicolumn{3}{c|}{Parameters} 
    & 
    \multicolumn{3}{c}{Consumption $c$ of \eqref{eq:mfg-single-stock-optimal-consumption}}
    \\
    \midrule
        $\kappa$  & $\beta$ & $\lambda -\beta$  &\,  optim.~candidate \, &\,  admissibility\,  & $c_t$ as $t\to \infty$
         \\
        \midrule
            \multirow{5}{*}{$<0$} & \multirow{3}{*}{$>0$} & $>0$ & \multirow{5}{*}{\cmark} & \multirow{5}{*}{\cmark} & $\searrow \beta$ \\
            \cline{3-3} \cline{6-6}
            & & $<0$ & & & $\nearrow \beta$ \\
            \cline{3-3} \cline{6-6}
            & & $=0$ & & & constant throughout\\
            \cline{2-3} \cline{6-6}
            &$=0$ & $>0$ & & & $\searrow 0$\\
            \cline{2-3} \cline{6-6}
            &$<0$ & $>0$ & & & $\searrow 0$\\
        \midrule
            \multirow{5}{*}{$>0$} & 
            \multirow{3}{*}{$>0$} & $>0$ &\cmark & \xmark & finite-time blow-up\\
            \cline{3-6}
            & & $<0$ & \cmark & \cmark & $\searrow 0$ \\
            \cline{3-6}
            & & $=0$ & \cmark & \cmark & constant throughout \\
            \cline{2-6}
            &$=0$ & $>0$ & \cmark & \xmark & finite-time blow-up \\
            \cline{2-6}
            &$<0$ & $>0$ & \cmark & \cmark & finite-time blow-up \\
        \midrule
        $=0$ & any & any & \cmark & \cmark & constant throughout\\
        \bottomrule
    \end{tabular*}
        \hfill
\end{table}

\subsubsection*{The environment's influence on the agents} 

One may wonder about the environment's impact on the agent in the particular case of EIS being close to zero.
When the environment hardly distinguishes the long-run preference for consumption from the one for wealth ($\kappa \to 0$), the agent who wants to faster reach a constant level of consumption needs to adapt by changing her risk and competition parameters in order to increase the absolute value of $\beta$.
For example, assume further no-discounting and zero risk-free interest rate, e.g., $\rho= r =0$, a highly risk-tolerant market with $\bE[\delta] \to \infty$ and mutually independent types $\theta$ and $\delta$, hence from \eqref{eq:beta-lambda-mfg-single-stock} we have $\theta_{\textrm{crit}} = \bE[\theta]$. Then, the agent changes in the following ways: the agent avoids the competition but accepts more risk ($\theta \to 0,~\delta > 1$); or competes at the average market level and accepts even more risk ($\theta \to \bE[\theta],~\delta \gg 1$); or accepts being highly competitive, but risk-averse ($\theta > \bE[\theta],~ \delta < 1$). 

Lastly, we note that $\beta$, as a quadratic function of $\delta_{\text{eff}}$, is bounded from above by $\frac{\mu^2}{8\sigma^2}$  but is not bounded from below. That means that under a mild \emph{substitution} regime ($\kappa \to 0^-$), the rational agent cannot outperform the upper-bound effective rate of portfolio return $\beta$ (shown in Fig.\ref{fig:plots-of-f-g-for-different-kappas}a and Fig.\ref{fig:plots-of-c-for-various-parameters}b). Nevertheless, by adjusting the risk-competition parameters $(\delta,\theta)$, the agent can decrease her consumption to zero as fast as he wants to. Notably, the choice of $\rho$ helps the agent adjust the magnitude of $\beta$ and hence the consumption preferences.
    
At this point, the market highly influences the choices of the player. Suppose the agent wants to maintain a desired level of consumption. In that case, she can do so by re-setting her risk-competition preferences (and in accordance to remain optimal ``as a martingale'').

\section{Open questions}

From the construction we provided, here and in the much simpler case \cite{platonov2020forward}, it is still open how to tackle the full MFG generalisation to market models featuring random coefficients. Here the mean field aggregation approach created in \cite{LackerZariphopoulou2017} is not possible anymore and a new tool is needed. We refer to Remark \ref{rem:open-question-after-nash-n-player} regarding the open issue of allowing agent to choose different $\kappa$ parameters.

Generalising the dynamics of the forward performance utility map \eqref{eq:SDEforForwardUtilityField-consumption} or \eqref{eq:StructureQgeneral} outside Assumption \ref{ass:regularity-u-consumption-mfg} to a fully It\^o-dynamics and stochastic strategies is also open. A crucial tool for such would be a general It\^o-Wentzell-Lions chain rule as developed in \cite{Platonov2019ito} and would likely build along \cite{KarouiMrad2013} or \cite{MusielaZariphopoulou2010}; or alternatively, an approach similar to \cite{chong2018optimal} can be taken where the volatility of the FPP is exogenously postulated as the agent's preferences (Forward-Backward SDEs is the tool of choice there). It is also open exploring game competition in the setting of  \cite{chong2018optimal} either for the finite-player game or the MFG as in  \cite{geng2017passive,AnthropelosGengZariphopoulou2020,platonov2020forward}. As pointed in \cite[Section 5]{platonov2020forward}, many other questions can be posed in this context of mean field FPP, ranging from possible non-solvability \cite{FreiDosReis2011}, to risk-sharing pricing \cite{BielagkLionnetDosReis2017}, indifference-pricing \cite{leung2012forward,chong2019pricing}, 
ergodic problems \cite{GechunZariphopoulou2016ergodic}.
The cases and analysis of \cite{AnthropelosGengZariphopoulou2020} can be extrapolated to the MFG case as well.

A particular case of interest playing to the strengths of the forward performance process framework, which we do not explore here, is the \textit{dynamic model selection} problem. FPPs allow for a dynamic update of market parameters to accommodate a switch in the market environment or the agent's perception of risk, which the standard utility theory does not allow. \cite[Section 4.4]{platonov2020forward} carries over to the construction we have provided here -- such would allow the agent to update the market-risk relative consumption preference $\kappa$. An alternative view to this problem, exploiting FPPs and time-dependent risk parameters, can be found in \cite{strub2018evolution}. 

Lastly, a comparative study between the capabilities of forward utility maps versus Epstein-Zin preferences within the Kreps-Porteus recursive utilities is, to the best of our knowledge, an open question. As shown, the forward performance criteria features the crucial property of the recursive utility, and hence this work stands to span new economics studies on EIS \cite{aase2016recursive,cundy2018essays,thimme2017intertemporal}.

\section*{Acknowledgments}
The authors thank T.~Zariphopoulou (U.~of Texas, US), M.~Anthropelos (U.~of Piraeus, GR), M.~Mrad (U.~Paris 13, FR) and D.~Lacker (Columbia U., US) for the helpful discussions.






\newpage
\appendix

\section{Proof of Theorem \ref{theo:nPlayerForwardNashGame-consumption}}
We prove Theorem \ref{theo:nPlayerForwardNashGame-consumption} in full detail.
\label{supsec:proof-N-game}
\begin{proof}

\emph{Step 1. Finding the investment strategy.}

First, we deal with the investment policy.
Injecting the condition $U_x/U_{xx}=-\delta_ix$ in \eqref{eq:optimstrategy-consumption} leads to the system ($i\in\{1,\dots,n\}$)
\begin{align*}
\pi^{i,*}_t 
&=
\frac1{{\nu_i^2+\sigma_i^2 \big(1+\frac{\theta_i(1-\delta_i)}{n-1}\big) }}
\Big(\theta_i \sigma_i (1-\delta_i)\frac n{n-1} \overline{\sigma\pi_t}  +  \mu_i \delta_i \Big), \quad \text{where} \quad \overline{\sigma\pi_t} =\frac1n\sum_{k=1}^n  \sigma_{k}\pi^{k,*}_t.
\end{align*}
The last identity expresses $\pi^{i,*}$ as a function of the unknown $\overline{\sigma\pi_t}$. To determine it we multiply both sides of the $\pi^{i,*}$ expression by $\sigma_i$ and average over $i\in\{1,\dots,n\}$. This yields the solvability condition
\begin{align}
\label{eq:pisigma-nplayers}
\overline{\sigma\pi_t}
=
\overline{\sigma\pi_t}
\psi^\sigma_n
+
\varphi^\sigma_n
\quad \Leftrightarrow\quad
\overline{\sigma\pi}
=\frac{\varphi^\sigma_n}{1-\psi^\sigma_n}\quad \textrm{as long as }\quad \psi^\sigma_n\neq 1.
\end{align}
Plugging the expression $\overline{\sigma\pi}$ into that for $\pi^{i,*}$ yields the result. That the optimal strategies are constant is now obvious and one finds that the corresponding $\eta$ (see \eqref{eq:rho-n-player}) is time-independent with $\eta_t^i = \eta_i$. 

If $\psi^\sigma_n=1$, then there exists no Nash equilibrium.

\emph{Step 2. Finding the explicit form of $\eta_i$.}

Just like for $\overline {\sigma\pi_t}$, we obtain an expression for $\overline {\mu\pi_t}=\frac1n\sum_{i=1}^k\mu_k\pi_t^{k,*}$ by multiplying $\pi^{i,*}$ by $\mu_i$ on both sides and averaging over $i$. We have 
\begin{align}
\label{eq:pimu-nplayers}
\overline {\mu\pi_t } = \frac n {n-1}\frac{\varphi^\sigma_n}{1-\psi^\sigma_n} \psi^\mu_n + \phi^\mu_n, 
\qquad
\overline{\mu\pi_t}^{(-i)} 
	= \frac n{n-1} \overline {\mu\pi_t } - \frac1{n-1}\mu_i\pi^i_t, 
\end{align}
where $\varphi^\mu_{n}$, $\psi^\mu_{n}$ are defined as 
\begin{equation*}
\varphi^\mu_{n} 
=\frac{1}{n}\sum_{k=1}^{n} \delta_{k} \frac{\mu^2_{k}}{\nu_{k}^{2}+\sigma_{k}^{2}(1+\frac{\theta_{k}(1-\delta_k)}{n-1})} 
, \qquad
\psi^\mu_n 
= \frac{1}{n-1}\sum_{k=1}^{n}\theta_{k}(1-\delta_k)\frac{\mu_k \sigma_{k}}{\nu_{k}^{2}+\sigma_{k}^{2}(1+\frac{\theta_{k}(1-\delta_k)}{n-1})}. 
\end{equation*} 
Similarly, defining 
\begin{align}
\label{eq:pinu-nplayers}
\overline{(\nu\pi_t)^2}=\frac1{n}\sum_{k=1}^n (\nu_k\pi^k_t)^2
\quad\Rightarrow \quad
\overline{(\nu\pi_t )^2}
=
\frac{1}{n} \sum_{k=1}^n 
\Big(\dfrac{\nu_k\theta_k \sigma_k \cdot \frac n{n-1} \cdot \frac{\varphi^\sigma_n}{1-\psi^\sigma_n}  +  \nu_k\mu_k \delta_k }{{\nu_k^2+\sigma_k^2 \big(1+\frac{\theta_k(1-\delta_k)}{n-1}\big) }}\Big)^2,
\end{align} 
together with $\overline{(\nu\pi_t)^2}^{(-i)}=\frac {n}{n-1} \overline {(\nu\pi_t)^2} - \frac1{n-1} (\nu_i\pi^i_t)^2$. 
Finally $\overline{\Sigma \pi^2_t} = \overline{(\nu\pi_t)^2} + \overline{(\sigma\pi_t)^2}$, and like for \eqref{eq:pinu-nplayers} we have 
\begin{align}
\label{eq:piSigma-nplayers}
    \overline{(\sigma\pi_t)^2} & = \frac{1}{n} \sum_{k=1}^n \bigg(\dfrac{\sigma_k\theta_k \sigma_k  \frac n{n-1} \frac{\varphi^\sigma_n}{1-\psi^\sigma_n}  +  \sigma_k\mu_k \delta_k }{{\nu_k^2+\sigma_k^2 \big(1+\frac{(1-\delta_k)\theta_k}{n-1}\big) }}\bigg)^2 
    \\ \nonumber 
    & \qquad  \text{ with }\ 
    \widehat{\Sigma \pi^2_t} = \frac{n}{n-1}\overline{\Sigma \pi^2_t} - \frac1{n-1}(\nu_i^2+\sigma_i^2)(\pi^i_t)^2.
\end{align}
Replacing these expressions in that of the constant $\eta^i_t = \eta_i$ in Equation \eqref{eq:lambda-consumption} of Section \ref{sec:FPPinitapowerrisk} one obtains the sought expression (in the Theorem's statement).

\emph{Step 3. Finding the consumption strategy and the relative performance utilities.}

The system of Equations \eqref{eq:system-equations-n-player-consumption} under Hypothesis \ref{ass:form-of-g-n-player}, i.e., $g_i(t) = f_i(t)^{1-\kappa}$, becomes
\begin{align}
\label{eq:system-equations-n-player-consumption-rewritten}
\begin{cases}
    c^i_t = \epsilon_i^{-\delta_i}\big(\tilde c_t^{(-i)}\big)^{\theta_i(1-\delta_i)} f_i(t)^{-\kappa\delta_i}, \\
    f^\prime_i(t) + \Big(\eta_i + \theta_i \big(1-\frac1{\delta_i}\big)\bar{c}_t^{(-i)}\Big) f_i(t) + \frac{\epsilon_i^{-\delta_i}}{\delta_i} \big(\tilde c_t^{(-i)}\big)^{\theta_i(1-\delta_i)}  f_i(t)^{1-\kappa\delta_i} = 0,
\end{cases}
\end{align}
where $\eta$ is given by \eqref{eq:rho-n-player}.
We mainly follow the machinery of the proof \cite[{Theorem 2.2}]{lackersoret2020many} to obtain the closed form solution.
We repeat these arguments and emphasise the notable discrepancies.
Substituting the LHS of the first equation of \eqref{eq:system-equations-n-player-consumption-rewritten} into the second, we obtain the linear ODE which solves as 
\begin{align*}
    f_i(t) = \exp \left(-\int_0^t \left(\eta_i + \theta_i\left(1-\frac1{\delta_i}\right)\bar{c}_s^{(-i)} + \frac1{\delta_i}c^i_s \right) ds\right).
\end{align*}
Now plugging it back into the first equation of \eqref{eq:system-equations-n-player-consumption-rewritten} we obtain
\begin{align*}
   c^i_t\exp\Big(-\kappa\int_0^t c^i_s ds\Big) &=\epsilon_i^{-\delta_i} \big(\tilde c_t^{(-i)}\big)^{\theta_i(1-\delta_i)} e^{-\kappa\delta_i\eta_it}\exp \Big(-\kappa (1-\delta_i)\theta_i\int_0^t \big( \bar{c}_s^{(-i)}\big) ds\Big). 
\end{align*}
After rewriting it with respect to $\tilde c_t = \Big(\prod_{k = 1}^n c^k_t\Big)^{\frac1n}$ and $\bar c_t = \frac1n\sum_{k=1}^n c^k_t$, we get
\begin{align}
    \label{eq:solving-nash-eq-consumption}
    c^i_t\exp\Big(-\kappa\int_0^t c^i_s ds\Big) &=\epsilon_i^{-\frac{\delta_i}{1 + \frac{\theta_i}{n-1}(1-\delta_i)}} \tilde c_t^{\frac{n}{n-1}\frac{\theta_i(1-\delta_i)}{1 + \frac{\theta_i}{n-1}(1-\delta_i)}} e^{\frac{-\kappa\eta_i\delta_i}{1 + \frac{\theta_i}{n-1}(1-\delta_i)}t}\\
    \nonumber
    &\qquad \qquad \times\exp \Big(-\kappa\frac{n}{n-1}\frac{\theta_i(1-\delta_i)}{1 + \frac{\theta_i}{n-1}(1-\delta_i)}\int_0^t \bar{c}_s ds\Big).
\end{align}
We take the geometric average of Equation \eqref{eq:solving-nash-eq-consumption} over $i = 1,\dots ,n$ to obtain
\begin{align*}
    \tilde{c}_t \exp\Big(-\kappa \int_0^t \bar{c}_sds\Big)  = \Big(\widetilde{\epsilon^{\delta}}\Big)^{-1} e^{\kappa\overline{\eta \delta}t} \bigg(\tilde{c}_t \exp\Big(-\kappa \int_0^t \bar{c}_sds\Big)\bigg)^{\overline{\theta(1-\delta)}},
\end{align*}
where
\begin{align*}
    & \widetilde{\epsilon^{\delta}} = \Bigg(\prod_{k = 1}^n \epsilon_k^{\frac{\delta_k}{1 + \frac{\theta_k}{n-1}(1-\delta_k)}}\Bigg)^{\frac1n}, 
    \quad 
    \overline{\eta\delta} = \frac1n \sum_{k = 1}^n \frac{\eta_i \delta_i}{1 + \frac{\theta_k}{n-1}(1-\delta_k)}, 
    \\
    & \qquad \textrm{ and }\quad 
    \overline{\theta(1-\delta)} = \frac1{n-1}\sum_{k = 1}^n \frac{\theta_k(1-\delta_k)}{1 + \frac{\theta_k}{n-1}(1-\delta_k)}.
\end{align*}

Hence we obtain that
\begin{align}
    \label{eq:consistency-geom-aver-c-n-player}
    \tilde{c}_t\exp\Big(-\kappa \int_0^t \bar{c}_sds\Big) = \Big(\widetilde{\epsilon^{\delta}}\Big)^{\frac{1}{\overline{\theta(1-\delta)}-1}} e^{\frac{-\kappa\overline{\eta\delta}}{\overline{\theta(1-\delta)}-1}t}.
\end{align}
Using the previous equation we rewrite \eqref{eq:solving-nash-eq-consumption} as 
\begin{align}
    \label{eq:solving-nash-eq-consumption-final-form}
   c^i_t\exp\Big(-\kappa\int_0^t c^i_s ds\Big) &= \lambda_i e^{-\kappa \beta_it},
\end{align}
where
\begin{align*}
    \lambda_i &= \epsilon_i^{-\frac{\delta_i}{1 + \frac{\theta_i}{n-1}(1-\delta_i)}} \Big(\widetilde{\epsilon^{\delta}}\Big)^{\frac{n}{n-1}\frac{\theta_i(1-\delta_i)}{{(\overline{\theta(1-\delta)}-1) (1 + \frac{\theta_i}{n-1}(1-\delta_i))}}},\\
    \beta_i &= \frac1{1 + \frac{\theta_i}{n-1}(1-\delta_i)}\bigg(\frac{n}{n-1}\frac{\theta_i(1-\delta_i)}{\overline{\theta(1-\delta)}-1} \overline{\eta\delta}-{\eta_i\delta_i}\bigg).
\end{align*}
Now consider two distinct cases $\kappa \neq 0$ and $\kappa = 0$.

\emph{Case 1: Let $\kappa \neq 0$.} Integrating \eqref{eq:solving-nash-eq-consumption-final-form} from $0$ to $t$ and taking logarithms we get
\begin{align}
    \label{eq:nash-equilibrium-integral-form-consumption-n-players}
    \kappa\int_0^t c_s^i ds = 
    \begin{cases}
        -\log\Big(1+ \frac{\lambda_i}{\beta_i}\big(e^{-\kappa\beta_it}-1\big)\Big) , & \beta_i \neq 0, \\
        -\log(1 - \lambda_i\kappa t), & \beta_i = 0.
    \end{cases}
\end{align}
Finally differentiating \eqref{eq:nash-equilibrium-integral-form-consumption-n-players} with respect to $t$, we obtain
\begin{align}
    \label{eq:nash-equilibrium-c-n-player}
    c_t^i =  
    \begin{cases}
        \bigg(\frac{1}{\beta_i} + \Big(\frac1{\lambda_i} - \frac1{\beta_i}\Big)e^{\kappa\beta_it} \bigg)^{-1}, & \beta_i \neq 0,\\
        \big(-\kappa  t + \frac1{\lambda_i}\big)^{-1}, & \beta_i = 0.
    \end{cases}
\end{align}

By direct inspection it is obvious that for certain combinations of parameters, namely $\kappa > 0,~ \beta_i < \lambda_i$ and $\kappa > 0,~\beta_i = 0$, the optimal consumption $c$ in not continuous and can even be negative. These cases are not admissible. Admissibility against parameter combinations is summarized in Table \ref{table:consumptionSignAsMapofParameters}.

Now, from the first line of \eqref{eq:system-equations-n-player-consumption-rewritten}, we get that
\begin{align}
    \label{eq:f-final-form-n-player}
    f_i(t) = \Big(\big(c_t^i\big)^{\frac1{\delta_i}}\big(\tilde{c}_t^{(-i)}\big)^{\theta_i(\frac1{\delta_i}-1)}\epsilon_i\Big)^{-\frac1{\kappa}} 
    \quad \text{and} \quad 
    g_i(t) = f_i(t)^{1-\kappa}.
\end{align}

\emph{Case 2: Let $\kappa = 0$.} Then \eqref{eq:solving-nash-eq-consumption-final-form} immediately yields $c_t^i = \lambda_i$ which is the same result if one  sets $\kappa=0$ in \eqref{eq:nash-equilibrium-c-n-player}.
\end{proof}

\section{Proof of Theorem \ref{theo:MFG-solution-consumption}}
\label{supsec:proofMFgame}

We prove Theorem \ref{theo:MFG-solution-consumption} in full detail.
\begin{proof}
We proceed stepwise in order to construct the constant mean field-equilibrium. To that end we must solve ii)-iii) in Definition \ref{def:MFG-Forward-problem-consumption} for given processes $\overline X,\overline \Gamma$  associated to some $(\pi,c)\in \cA_{\textrm{MF}}$. Condition iv) of the MF-equilibrium allows us to focus only on processes of the form $\overline X_t = \exp\bE[\log X_t| \cF^B_t]$ and $\overline \Gamma_t=\exp\bE[\log c_t| \cF^B_t]$ where $X$ solves \eqref{def:X-MFG-consumption} for some strategy $(\pi,c) \in \cA_{MF}$.

\emph{Step 0. The average wealth process.} To solve the above problem given $(\overline X_t)_{t\geqs 0}$ it suffices to restrict ourselves to processes $(\overline X_t)_{t\geqs 0}$ satisfying $\overline X_t = \exp\bE[\log X^\pi_t| \cF^B_t],~\bP$-a.s. 
We then have via It\^o's formula and the arguments from \cite{LackerZariphopoulou2017} $\bP$-a.s.
\begin{align*}
\nonumber
\overline X_t 
& 
= \exp \bE[\log X_t| \cF^B_t]
\\
&
=  \exp\bE\Big[ \log\xi + \int_0^t \big(\mu \pi_s - \frac12 \pi_s^2(\nu^2+\sigma^2)\big) ds 
\\
&\hspace{4cm}+ \int_0^t\nu \pi_s dW_s + \int_0^t\sigma \pi_s dB_s - \int_0^t c_s ds\Big| \cF^B_t\Big]
\\
&
= \exp \Big[\overline{\log\xi} + \int_0^t \big(\overline{\mu \pi_s}-\frac12\overline{\Sigma \pi_s^2}\big)  ds + \int_0^t \overline{\sigma \pi_s} dB_s - \int_0^t\overline{c}_sds\Big]\\
&
= \overline\xi+  \Big( \int_0^t\eta\overline X_s ds +\int_0^t\overline{\sigma \pi_s}\overline X_s dB_s - \int_0^t \bar{c}_s\overline X_s ds \Big),
\end{align*}
where, for consistency of notation with respect to the previous section, we denote
\begin{align*}
\eta &:= \overline {\mu\pi_s} - \frac12(\overline{\Sigma\pi_s^2} - \overline{\sigma\pi_s}^2), \ \
\overline \xi:=\exp\bE[ \log\xi],
\ \
\overline{\mu \pi_s}:=\bE[\mu \pi_s] 
,\ \
\overline{\sigma \pi_s}:=\bE[ \sigma \pi_s], \quad \bar{c} = \bE[c].
\end{align*}
Hence, for $(\pi,c) \in \cA^{\textrm{MF}}$ we can define the process $Z^{\pi,c} = X^{\pi,c}  \overline{X^{\pi,c}}^{-\theta}$. By It\^o's formula we derive its SDE dynamics as   
\begin{align*}
\frac{d Z^{\pi,c}_t}{Z^{\pi,c}_t} 
&= \big(\mu \pi_t-\theta \overline{\mu \pi_t} + \frac{\theta}2\overline{\Sigma\pi_t^2} + \frac{\theta^2}2 \overline{\sigma\pi_t}^2 -\theta\sigma\pi_t \overline{\sigma\pi_t} \big)dt + \nu \pi_t d W_t +\big(\sigma \pi_t-\theta \overline{\sigma \pi_t} \big)dB_t,\\
&\qquad-(c_t-\theta\bar{c}_t)dt,
\quad Z^{\pi,c}_0=\xi (\overline \xi)^{-\theta}.
\end{align*}
We proceed to solve the MFG Forward performance problem of Definition \ref{def:MFG-Forward-problem-consumption} with its help.

Applying It\^o's formula to $U(Z^{\pi,c}_t,t)$ yields
\begin{align}
\label{eq:generalSDEafterItoWentzell-MFG-consumption}
\nonumber
 d Q(Z^{\pi,c}_t ,t) 
&= U_t(Z^{\pi,c}_t ,t)dt + U_x(Z^{\pi,c}_t ,t) d Z^{\pi,c}_t + \frac12 U_{xx}(Z^{\pi,c}_t ,t) d \langle Z^{\pi,c}_t \rangle + V(Z^{\pi,c}_t ,t)dt
\\ \nonumber
& 
=\Big[U_t(Z^{\pi,c}_t ,t)   
+ U_x(Z^{\pi,c}_t ,t) \big(\mu \pi_t - \theta \overline{\mu \pi}+ \frac{\theta}2\overline{\Sigma\pi^2} + \frac{\theta^2}2 \overline{\sigma\pi}^2 -\theta\sigma\pi_t \overline{\sigma\pi}\big) Z^{\pi,c}_t\\
&
\qquad +\frac12 U_{xx}(Z^{\pi,c}_t ,t) \Big( (\nu \pi_t )^2 + \big(\sigma \pi_t - \theta \overline{ \sigma \pi } \big)^2\Big)(Z^{\pi,c}_t)^2\Big]dt
\\ \nonumber
&
\qquad
+ U_x(Z^{\pi,c}_t ,t) \nu\pi_t Z^{\pi,c}_t dW_{t} 
+ U_x(Z^{\pi,c}_t ,t)\big(\sigma \pi_t- \theta \overline{ \sigma \pi} \big)Z^{\pi,c}_t dB_t\\
\nonumber
&
\qquad
+ U_x(Z^{\pi,c}_t ,t)(c_t - \theta \bar{c}_t)dt + V(\hat{c}_tZ^{\pi,c}_t,t)dt,
\end{align}
with $U(Z^{\pi,c}_0 ,0)=U(\xi\big(\bar \xi)^{-\theta},0\big)= \frac1{1-\frac1{\delta}}\big(\xi(\bar \xi)^{-\theta}\big)^{1-\frac1{\delta}} $ and using that $B,W$ are i.i.d.

\emph{Step 1. The candidate best responses strategies $\pi^*,c^*$.} As before, the process $U(Z^{\pi,c}_t,t)$ becomes a Martingale at the optimum $\pi$. Direct computations using first order conditions ($\partial_{\pi} \textrm{``drift''}=\partial_{c} \textrm{``drift''}=0$) yield
\begin{align}
    \nonumber
    &\begin{cases}
        0 + U_x \cdot \big( \mu -0 -\theta\sigma\overline{\sigma \pi_t}\big)Z^{\pi,c}_t + \frac12 U_{xx} \Big(2 \pi \nu^2 + 2\big(\sigma \pi_t- \theta \overline{\sigma \pi_t} \big) \sigma \Big)\big(Z^{\pi,c}_t\big)^2 = 0,\\
        -U_x(Z^{\pi,c}_t,t) Z^{\pi,c}_t + V_x(\hat{c}_t Z^{\pi,c}_t,t)\frac{Z^{\pi,c}_t}{\big(\overline{c}^{(-i)}_t\big)^{\theta_i}} = 0,
    \end{cases}
    \\
    \label{eq:best-responses-mfg}
     & \qquad \qquad \qquad\Rightarrow\quad 
    \begin{cases}
        \pi_t &= \frac1{\nu^2+\sigma^2}\Big(\theta \sigma \overline{\sigma\pi_t} +  (\mu -\theta\sigma\overline{\sigma \pi_t})\dfrac{U_x}{U_{xx}Z^{\pi,c}_t}\Big),
        \\
        c_t &= \dfrac{(V_x)^{-1}\Big(U_x(Z^{\pi,c}_t,t) \tilde{c}_t^{\theta},t\Big)\tilde{c}_t^{\theta}}{Z^{\pi,c}_t}.
    \end{cases}
\end{align} 
Now we inject into the first equation the CRRA constraint $U_x/U_{xx}=-\delta x$
and use Hypothesis \ref{ass:form-of-U-V-mean-field-g-with-f} to obtain
\begin{align*}
\pi_t =\frac1{\nu^2+\sigma^2}\Big(\theta \sigma \overline{\sigma\pi_t} +  (\mu -\theta\sigma\overline{\sigma \pi_t}) \delta\Big)
\quad\textrm{and}\quad
        c_t = \tilde c_t^{\theta(1-\delta)} {f(t)}^{-\delta\kappa}.
\end{align*}
By inspection, it is clear that $\pi^*$ is a $\cF_0^{\textrm{MF}}$-measurable RV which is independent of time and is well defined as long as $\overline{\sigma\pi}$ is finite. The derivation of the closed form of the optimal consumption needs further work and is carried out further below.


\emph{Step 2. The optimality of the strategy.} 
In contrast to the $n$-player optimisation the mean field game is defined with reference to a pair of average processes $\overline X_t$ and $\overline \Gamma_t$ against which the equilibrium is defined through a fixed-point stationarity identity. We provide a verification procedure similar to that in \cite[Proof of Theo.~2]{platonov2021forward-published}.
The original constant strategy $\pi$ is a MF-equilibrium if and only if for all $t\geqs 0,~\bP$-a.s.
\begin{align*}
    &\begin{cases}
        \bE[\log X^{\pi,c}_t| \cF^B_t] = \bE[\log X^{\pi^*,c^*}_t| \cF^B_t],\\
        \bE[\log c_t| \cF^B_t] = \bE[\log c^*_t| \cF^B_t],
    \end{cases}\\ 
    \Leftrightarrow \quad
    &\begin{cases}
        \overline{\log\xi} + \int_0^t\big(\overline{\mu \pi_s} - \frac12 \overline{\Sigma \pi_s^2} \big)ds + \int_0^t\overline{\sigma \pi_s} dB_s - \int_0^t\bar{c}_sds 
        \\
        \hspace{4cm}=
        \overline{\log\xi} + \int_0^s\big(\overline{\mu \pi^*_s}-\frac12\overline{\Sigma (\pi_s^*)^2}\big) ds + \int_0^t\overline{\sigma \pi_s^*} dB_s-\int_0^t\overline{c^*_s}ds,\\
        \tilde{c}_t = \tilde{c}^*_t,
\end{cases}
\end{align*}
where we denote $\tilde{c}_t := \exp\bE[\log c_t| \cF^B_t] = \exp\bE[\log c^*_t| \cF^B_t] =: \tilde{c}^*_t$.
After taking expectations in the first equation it follows that $\pi$ is a MF-equilibrium if and only if the following three conditions hold $\bP$-a.s.
\begin{align}
\label{eq:equilibrium-big-condition-SM}
    \begin{cases}
        \overline{\sigma \pi_t}
        =\overline{\sigma \pi_t^*},
        \\
        \int_0^t \big(\overline{\mu \pi_s} -\frac12\overline{\Sigma \pi_s^2} \big) ds\ - \int_0^t \overline{c_s}ds  =\int_0^t\big(\overline{\mu \pi_s^*} -\frac12\overline{\Sigma (\pi_s^*)^2}\big) ds - \int_0^t \overline{c^*_s}ds,\\
        \tilde{c}_t  = \tilde{c}^*_t.
    \end{cases}
\end{align}
Using \eqref{eq:best-responses-mfg} (with $U_x/U_{xx}=-\delta x$ replaced in) one derives (using the expressions $\varphi^\sigma,\psi^\sigma$)
\begin{align*}
\sigma \pi_t^* & =  \theta(1-\delta)\frac{ \sigma^2}{\nu^2+\sigma^2} \overline{\sigma \pi_t} +  \delta \frac{\mu\sigma }{\nu^2+\sigma^2}
\qquad  \Rightarrow  \quad 
\overline{\sigma \pi^*}
= \overline{\sigma \pi_t} \psi^\sigma + \varphi^\sigma.
\end{align*}
Using that $\overline{\sigma \pi_t}=\overline{\sigma \pi_t^*}$ yields solvability if $\psi^\sigma=\bE\big[\theta(1-\delta)\frac{ \sigma^2}{\nu^2+\sigma^2}\big]\neq 1$. Thus
\begin{align}
\label{eq:sigma-pi-pverline}
    \overline{\sigma \pi^*}=\overline{\sigma \pi}=\frac{\varphi^\sigma}{1-\psi^\sigma}=\textrm{constant},
\end{align}
and the $\pi^*$ expression of \eqref{eq:mfg-equilibrium-pi-c} follows. To exploit the next condition, we solve PDE \eqref{eq:MF-SPDE-consumption} under Hypothesis \ref{ass:form-of-U-V-mean-field-g-with-f}. Together with the optimal candidate consumption we have
\begin{align}
    \label{eq:mfg-solving-equilibrium-system-SM}
    \begin{cases}
        c^*_t = \epsilon^{-\delta}(\tilde c_t)^{\theta(1-\delta)} f_i(t)^{-\kappa\delta}, \\
        f^\prime(t) + \Big(\chi - \theta \big(1-\frac1{\delta}\big)\big(\overline{\mu\pi_t} -  \frac12\overline{\Sigma \pi_t^2} - \bar{c}_t \big)\Big) f(t) + \frac{\epsilon^{-\delta}}{\delta} (\tilde c_t)^{\theta(1-\delta)}  f(t)^{1-\kappa\delta} = 0,
    \end{cases}
\end{align}
with
\begin{align*}
    \chi = \Big(1-\frac1{\delta}\Big) \bigg(\frac{\delta\big(\mu-\theta\sigma\overline{\sigma\pi^*}(1-\frac{1}{\delta})\big)^2}{2(\nu^2+\sigma^2)} + \frac{\theta^2}{2} \big(\overline{\sigma\pi^*}\big)^2 \big(1-\frac1{\delta}\big) - {r (1-\theta)} \bigg) - {\rho}.
\end{align*}
Plugging the first equation of \eqref{eq:mfg-solving-equilibrium-system-SM} into the second one, we solve for $f$ to obtain 
\begin{align*}
    f(t) = \exp\Big(- \int_0^t \Big(\chi - \theta\big(1-\frac1{\delta}\big) \big(\overline{\mu\pi_s} -  \frac12\overline{\Sigma \pi_s^2} - \bar{c}_s \big) + \frac1\delta c^*_s \Big)ds\Big).
\end{align*}
Now we substitute it back into the first equation of \eqref{eq:mfg-solving-equilibrium-system-SM} to get
\begin{align*}
\nonumber
    c^*_t \exp\big(-\kappa\int_0^tc_s^*ds\big)
    &= \epsilon^{-\delta}(\tilde c_t)^{\theta(1-\delta)}e^{\chi\delta\kappa t}
    \\ 
    &\qquad \times \exp\Big(\theta(1-\delta)(\kappa-1)\int_0^t \big(\overline{\mu\pi_s} -  \frac12\overline{\Sigma \pi_s^2} - \bar{c}_s \big)ds\Big).
\end{align*}
Now we substitute according to second equilibrium identity of \eqref{eq:equilibrium-big-condition-SM} to obtain
\begin{align}
    \label{eq:mfg-solving-c-step2-SM}
    & c^*_t \exp\big(-\kappa\int_0^tc_s^*ds\big)
    = \epsilon^{-\delta}(\tilde c_t)^{\theta(1-\delta)}e^{\eta\delta\kappa t} \exp\big(-\theta(1-\delta)\kappa\int_0^t \overline{c_s^*}ds\big),
\end{align}
where
\begin{align*}
    \eta_t = \chi - \theta\big(1-\frac1\delta\big)\bigg(\overline{\mu\pi_t^*} -  \frac12\overline{\Sigma (\pi_t^*)^2}\bigg).
\end{align*}
Taking the logarithm, expectation and exponent on both sides of \eqref{eq:mfg-solving-c-step2-SM} we get
\begin{align*}
    \tilde{c}^*_t \exp\big(-\kappa\int_0^t\overline{c_s^*}ds\big)= \big(\widetilde{\epsilon^{\delta}}\big)^{-1}e^{\kappa\overline{\eta\delta}t} \Big(\tilde c_t\exp\big(-\kappa\int_0^t \overline{c_s^*}ds\big)\Big)^{\overline{\theta(1-\delta)}},
\end{align*}
where $\widetilde{\epsilon^{\delta}} = \exp\bE[\delta\log \epsilon]$, $\overline{\eta\delta} = \bE[\eta \delta]$ and $\overline{\theta(1-\delta)} = \bE[ \theta(1-\delta)]$.
Using the last equilibrium identity,  $\tilde{c}_t = \tilde{c}^*_t$, we rewrite the expression as
\begin{align*}
    \tilde{c}^*_t \exp\big(-\kappa\int_0^t\overline{c_s^*}ds\big)= \widetilde{\epsilon^{\delta}}^{\frac1{\overline{\theta(1-\delta)}-1}} e^{-\frac{\overline{\eta\delta}}{\overline{\theta(1-\delta)}-1}\kappa t}.
\end{align*}
Plugging it back into \eqref{eq:mfg-solving-c-step2-SM} we obtain
\begin{align*}
    c^*_t \exp\big(-\kappa\int_0^tc_s^*ds\big) = \lambda e^{-\kappa \beta t},
\end{align*}
where
\begin{align*}
     \lambda&=\epsilon^{-\delta}\big(e^{\bE[\delta\log \epsilon]}\big)^{\frac{\theta(1-\delta)}{\bE[ \theta(1-\delta)]-1}}, 
    \qquad 
    \beta = \frac{\theta(1-\delta)}{\bE[ \theta(1-\delta)]-1}\bE[\eta \delta] -\eta\delta.
\end{align*}
The same arguments used in the proof of Theorem \ref{theo:nPlayerForwardNashGame-consumption} yield
\begin{align*}
    c^*_t =  
    \begin{cases}
        \bigg(\frac{1}{\beta} + \Big(\frac1{\lambda} - \frac1{\beta}\Big)e^{\kappa\beta t} \bigg)^{-1}, & \beta \neq 0,\\
        \big(-\kappa t + \frac1{\lambda}\big)^{-1}, & \beta = 0.
    \end{cases}
\end{align*}
As in Theorem \ref{theo:nPlayerForwardNashGame-consumption}, one finds that for certain combinations of parameters, namely $\kappa > 0,~ \beta < \lambda$ and $\kappa > 0,~\beta = 0$, the optimal consumption $c$ in not continuous and even can be negative. These cases are not admissible and admissibility against parameter combinations is summarized in Table \ref{table:consumptionSignAsMapofParameters}.

Finally, from the first line of \eqref{eq:mfg-solving-equilibrium-system-SM}, we get that
\begin{align*}
    f(t) = \Big(\big(c_t\big)^{\frac1{\delta}}\big(\tilde{c}_t\big)^{\theta(\frac1{\delta}-1)}\epsilon\Big)^{-\frac1{\kappa}} \quad \text{and thus} \quad g(t) = \Big(\big( c_t\big)^{\frac1{\delta}}\big(\tilde{c}_t\big)^{\theta(\frac1{\delta}-1)}\epsilon\Big)^{\frac{\kappa-1}{\kappa}}.
\end{align*}
Now we are left to conclude using the closed forms of $\overline{\mu\pi_t^*}$ and $\overline{\Sigma(\pi_t^*)^2}$.
First, we have
\begin{align}
\label{eq:ConsistencyExpressions-alphasigma-alphamu}
\overline{\mu \pi_t^*} 
=  \frac{\varphi^\sigma}{1 - \psi^\sigma} \psi^\mu + \varphi^\mu =\textrm{constant},
\end{align}
using $\psi^\mu,\varphi^\mu$.
Finally, we find the expression for $\overline{\Sigma\pi^2}$. Multiplying separately the $\pi^*$ of \eqref{eq:mfg-equilibrium-pi-c} by $\sigma$ and $\nu$, squaring, taking expectation and summing the results, we have
\begin{align}
    \label{eq:ConsistencyExpressions-Sigmapi}
    \overline{\Sigma\pi^2} = \bE \bigg[ \frac1{\nu^2+\sigma^2}\Big( \theta (1-\delta)\sigma \frac{\varphi^\sigma}{1-\psi^\sigma}  +  \mu \delta \Big)^2 \bigg].
\end{align}
The explicit form of the expression $\eta$ follows by injecting these identities in \eqref{eq:rho-mfg}.

We now address the non-solvability. If in the last equation of \eqref{eq:ConsistencyExpressions-alphasigma-alphamu} one has $\psi^\sigma=1$ and $\varphi^\sigma\neq 0$ then the equation has no solution and hence no strong MF-equilibrium exists.  The case  $\psi^\sigma=1$ and $\varphi^\sigma= 0$ is impossible. Since $\mu>0$ and $\delta>0$ by hypothesis, it implies that $\sigma=0$ and hence that $\psi^\sigma=0 $ contradicting the condition $\psi^\sigma=1$.

\emph{Step 3. The MFG forward performance process dynamics.} Injecting the consistency PDE \eqref{eq:MF-SPDE-consumption} in the expression for  $d U(Z^{\pi,c}_t ,t)$ given into \eqref{eq:generalSDEafterItoWentzell-MFG-consumption} yields 
\begin{align*}
    &
    d Q(Z^{\pi,c}_t ,t) 
    \\
    &
    = 
    U_x(Z^{\pi,c}_t ,t)  \Big(\nu\pi_tdW_{t} 
    + \big(\sigma \pi_t- \theta\frac{\varphi^\sigma}{1- \psi^\sigma} \big)dB_t\Big)Z^{\pi,c}_t
    \\
    &
    \quad \quad +
    \frac12 U_{xx}(Z^{\pi,c}_t ,t) \frac{1}{(\nu^2+\sigma^2)}  \bigg| \pi_t (\nu^2+\sigma^2)
    -\Big( \theta (1-\delta)\sigma  \frac{\varphi^\sigma}{1- \psi^\sigma} + \mu \delta \Big)\bigg|^2(Z^{\pi,c}_t)^2dt
    \\
    &
    \quad \quad + V(\hat c_t Z^{\pi,c}_t, t) - U_x(Z^{\pi,c}_t,t)c_tZ^{\pi,c}_t - \widetilde V(U_x(Z^{\pi,c}_t,t),t).
    \end{align*}
\end{proof}

\section{Regions of monotonicity of optimal consumption}

We enhance Fig.~3 (case $\kappa=1$) of the earlier contribution \cite{lackersoret2020many} to the context of our work and comment further on the finer interplay of $\kappa$ with $\beta$ and $\lambda$ on the consumption policy $c$ (see end of Section \ref{sec:examplesInterpretations}). 

In Figure \ref{fig:plots-of-monotonicity-c}, we have two pictures of the regions of monotonicity of $c_t$ for $\kappa>0$ and $\kappa<0$. We can see that having $\kappa$ across the region given by $\delta=1$ symmetrically reverses the direction of monotonicity for $c$. The region of consumption in the plot having a constant consumption regime has a constant color (we do not mark such level curves apart from the region boundaries).

    \begin{figure}[hbt!]
        \centering 
        \includegraphics[width=\textwidth]{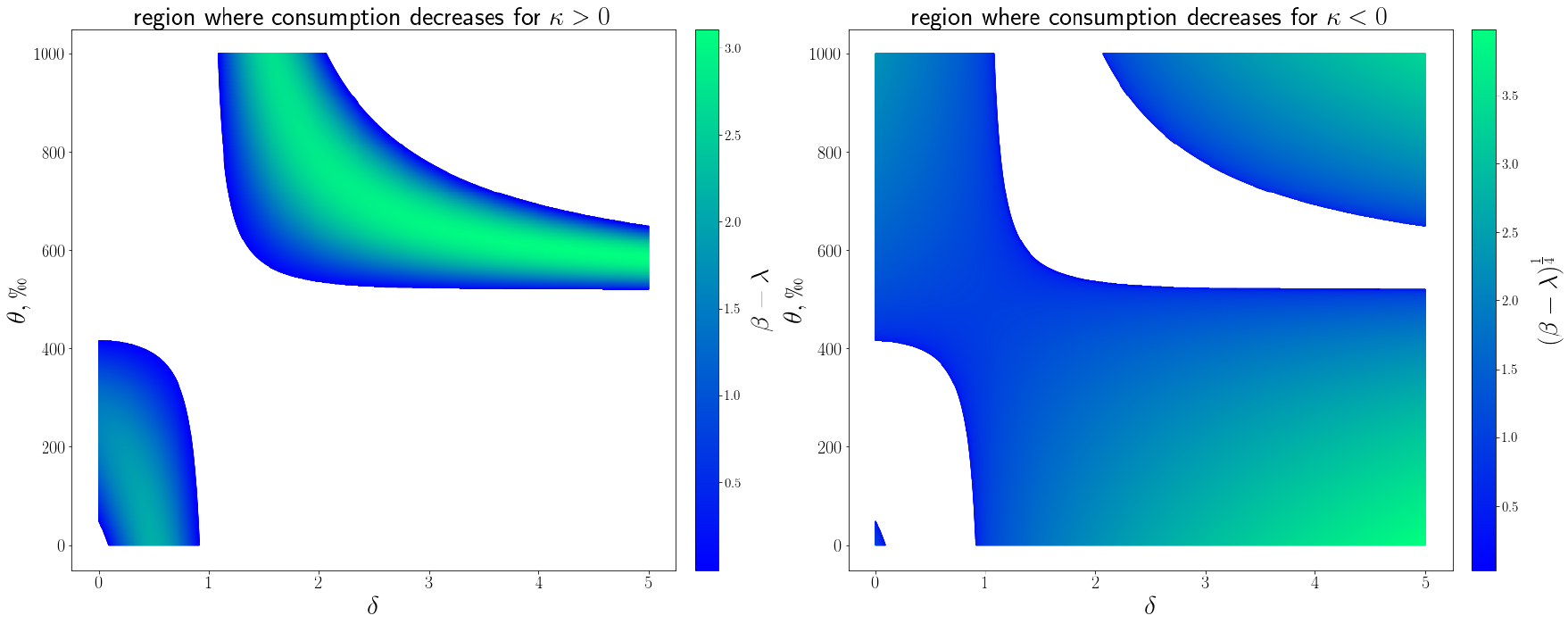}
        \caption{The regions of monotonicity for $c_t$ for $\kappa>0$ and $\kappa<0$ as function of $\delta$  and $\theta$ (in \textperthousand), the agent lying inside (outside) the coloured region decreases (increases) consumption rate over time. The agents on the border consume at a constant rate. The colour gradient relates to the speed of monotonicity characterized by $\beta-\lambda$ or a function of it. The set of parameters is taken from \cite[Fig.~3 (case $\kappa=1$)]{lackersoret2020many}, namely $\mu = 5,~ \sigma = 1,~ \epsilon = 1,~ \bE[\log \epsilon] = 0,~ \bE[\theta (1-\delta)] = -1.6,~ \bE[\delta] = 5,~ \theta_{\text{crit}} = 0.52,~{\rho = r = 0}$. 
        }
        \label{fig:plots-of-monotonicity-c}
    \end{figure}


%
%
%

\end{document}